\documentclass[a4paper,onecolumn,11pt,accepted=2020-05-13]{quantumarticle}
\pdfoutput=1
\usepackage[utf8]{inputenc}
\usepackage[english]{babel}
\usepackage[T1]{fontenc}
\usepackage{amsmath}
\usepackage{hyperref}

\usepackage{cite}
\usepackage{multirow}
\usepackage{physics}
\usepackage{color}
\usepackage{ulem}
\usepackage{enumerate}
\usepackage{bm, bbold}

\usepackage[caption=false]{subfig}

\usepackage[numbers]{natbib}
\usepackage{graphicx, xcolor}
\usepackage{braket}
\DeclareMathOperator{\lcm}{lcm}
\usepackage{amsthm}
\usepackage{amsfonts, amssymb}
\usepackage{bbm}
\newtheorem{prep}{Proposition}
\newtheorem{corol}{Corollary}

\begin{document}

\title{Energy upper bound for structurally-stable $N$-passive states}

\author{Raffaele Salvia}
\affiliation{Scuola Normale Superiore and University of Pisa, I-56127 Pisa, Italy}
\orcid{0000-0002-0006-7630}
\email{raffaele.salvia@sns.it}
\author{Vittorio Giovannetti}
\affiliation{NEST, Scuola Normale Superiore and Istituto Nanoscienze-CNR, I-56126 Pisa, Italy}

\maketitle

\begin{abstract}
		Passive states are special configurations of a quantum system which exhibit no energy decrement
at the end of an arbitrary cyclic driving of the  model Hamiltonian. 
When applied to  an increasing number of copies of the initial density matrix,
the requirement of passivity  induces a hierarchical ordering 
which, in the asymptotic  limit of infinitely many elements,  pinpoints 
ground states and thermal Gibbs states. In particular,  for large values of $N$ the energy content of a $N$-passive state which is also structurally stable (i.e. capable to maintain its passivity status under small perturbations of the model Hamiltonian), 
is expected to be close to the corresponding value of the thermal Gibbs state
which 
has the same entropy. In the present paper we provide a quantitative 
assessment of this fact, by producing an upper bound for  
the energy of an arbitrary $N$-passive, structurally stable state  which only depends on the spectral properties of the Hamiltonian of the system. 
We also show the condition under which our inequality can be saturated. A generalization of the bound is finally 
presented that, for sufficiently large $N$, applies to states which  are $N$-passive, but not necessarily structurally stable. 
\end{abstract}

	\section{Introduction}
One of the most striking differences between classical and quantum thermodynamics is that, while the properties of macroscopic system in equilibrium can be described with a small number of degrees of freedom, the unitary evolution prescribed by the laws of quantum mechanics has as many conserved quantities as the dimension of the Hilbert space \cite{ALICKIREV}. A closed quantum system can not ``lose its memory'' and relax to a thermal state. One of the consequences of this fact is that the amount of work that we can extract coupling a quantum system to a thermal bath (the free energy of the system) is, in general, larger than the work that we can extract from the system alone, called the \emph{ergotropy} of the system~\cite{PASSIVE3}. There are states of a quantum system which are not in thermal equilibrium, but that are nonetheless \emph{passive states}, in the sense that their energy can not decrease under unitary evolution.

The passive states of a quantum system $A$ 
consist of all density matrices that commute with the
system Hamiltonian $H$ and have no population inversions~\cite{PASSIVE1,PASSIVE2,ALICKIREV}. 
Originally  introduced in Ref.~\cite{PASSIVE1} by linking them to the 
Kubo-Martin-Schwinger   thermal stability condition~\cite{KUBO,MARSCH,BATTY},
passive states exhibit zero ergotropy, i.e. 
zero maximum mean energy decrement 
when forcing the system 
to undergo an unitary evolution induced by cyclic external modulations of $H$.
In the Kelvin-Planck formulation of the second law of thermodynamics, 
ergotropy
can be interpreted as  the maximum work that  can be extracted from a system~\cite{PASSIVE3,PASSIVE4},  suggesting the identification of passive states 
as a primitive form of thermal equilibrium. In view of this property, ergotropy and passive states play  a key role 
in quantum thermodynamics~\cite{REVQTHERMO,ALICKIREV}, where they help in clarifying several aspects of the theory,
spanning from foundational issues at the interplay between physics and information~\cite{LEWENSTEINrew,JANET,SPARC,BRANDAO,SKRZ,ACIN1,ACIN2,ERGO1,ERGO_BIP,BERA1}, to more practical issues, such as 
the characterisation of optimal thermodynamical cycles~\cite{ALICKIREV,KURI1,KURI2,KURI3,FRIIS,BINDER,ACIN} 
and the charging efficiency of quantum batteries models~\cite{BATTERYREV,PASSIVE4,BATTERY1,BATTERY2,BATTERY3,BATTERY4}. 
Passive states have been also identified as  
optimisers 
for several entropic functionals which are relevant in the theory of quantum communication~\cite{GIACOMO1,GIACOMO2,GIACOMO3},
and as suitable generalizations of the
vacuum state for  quantum field theory in curved space-time models~\cite{CURVED}.

A natural generalization of passivity can be obtained by considering  multiple copies of the original system~\cite{PASSIVE1,PASSIVE2}.
In particular, a density matrix  $\rho$  of $A$ is said to be $N$-passive with respect to the
local Hamiltonian $H$ when, 
given $N$ identical copies of it, one has that $\rho^{\otimes N}$ is passive when considering 
as joint Hamiltonian of the compound the sum of $N$ copies of $H$.
It turns out that $N$-passive
states are also $N'$-passive for all $N'\leq N$, the opposite inclusion not be granted in general, inducing a strict hierarchical ordering 
on the associated sets.
In this framework thermal Gibbs states 
share the exclusive property of being the only density matrices of the system 
which are {\it completely passive},  i.e.  passive at all order $N$, and
also being structurally stable~\cite{PASSIVE1,PASSIVE2,PASSIVE4,SKRZ}. Structural stability ensures  that the state under consideration will
remain passive even when the system Hamiltonian undergoes small perturbations. This condition is naturally granted to all passive configurations when $H$ has a non-degenerate
spectrum, but becomes a non trivial requirement in the presence of degeneracies. 
A direct consequence of the above mentioned property of Gibbs states is that, for $N$ large enough, the mean energy  $E(\rho;H)$ of a structurally-stable, $N$-passive density matrix $\rho$ must approach the mean energy  $E_{\beta(\rho)}(H)$ of the Gibbs configuration $\omega_{\beta(\rho)}$ that has the same entropy of $\rho$ -- the latter being always
a lower bound for $E(\rho;H)$, i.e. 
$E_{\beta}(H)\leq E(\rho,H)$.
Aim of the present work is to investigate how the gap between $E(\rho;H)$  and $E_{\beta}(H)$ reduces as
$N$ increases. For this purpose we prove an inequality which provides an upper bound for  $E(\rho;H)$ in term of $E_{\beta}(H)$, via a multiplicative factor which only depends upon the spectral properties of the Hamiltonian,
and which converges asymptotically to 1 as $N$ increases. 
This allows us to provide a quantitative estimation of the way in which a quantum effect (the gap between ergotropy and free energy) decreases when the size of the system (quantified by the number $N$ of copies) increases, and disappears in the macroscopic limit $N \to \infty$.

Our work can be seen as a generalization of Ref.~\cite{ACIN}, which charachterises the most energetic 1-passive states, to the case of $N$-passive states.
Incidentally, following the same argument presented in Ref.~\cite{ACIN}, our findings can also be used to 
give a lower bound for the work that can be extracted from a system, hence providing a practical tool to estimate the usefulness of a given state from the perspective of average work extraction.
Furthermore, our results could be used to derive lower bounds for the Carnot efficiency of fully quantized heat engines \cite{KURI2}, in the cases where the quantized piston is modelled as a product of $N$ identical systems. 

We stress that the derivation presented here relies heavily on the structural stability property of the input states; if we lift
such condition, the bounds do not apply in general. However, for sufficiently large values of $N$, we also give  
a variant of the inequality which remains true for all $N$-passive states (not
necessarily structurally stable) -- 
see Table~\ref{TABLE1} for a summary of the results of this paper.

	\begin{table*}[t!]
		\centering
		\resizebox{1.0\textwidth}{!}{
		\begin{tabular}{|c|c|c|c|}
	\hline
	Dimension & Spectrum of $H$ & Set of states  &   Energy bound   \\ \hline \hline
	$d=2$ & two-level  & $\rho\in {\mathfrak P}_H^{(1)}$     &    $E(\rho; H) = E_{\beta(\rho)}(H)$ \\ \hline 
	$d\geq 3$ & two-level &  $\rho\in 
	\bar{\mathfrak P}_H^{(1,1)}$  &   $E(\rho; H) = E_{\beta(\rho)}(H)$  \\ 
	&&  $\rho\in {{\mathfrak P}_H^{(N)}}^*$  &   $E(\rho;H)<  \frac{N-1}{N-2} E_{\beta(\rho)}(H)  +  \frac{(d_0 - 1)Z_{\beta(\rho)}^{-1}\epsilon_{\max}}{N-2}  +  \mathcal{O} \left( \frac{1}{N^2} \right)$  \\\hline  
	$d\geq 3$  & beyond two-level
	& \multirow{ 2}{*}{$\rho\in {\mathfrak P}_H^{(N)}$} &  \multirow{ 2}{*}{$E(\rho;H) 
		\leq E_{{\beta}({\rho})}(H) \min^{\star} \left\{ \left( 1 - \tfrac{R(H)}{N} \right)^{-1},e^{\beta({\rho}) \epsilon_{\max} \frac{R(H)}{N}} \right\}$} 
	\\ 
	& non-degerate & & \\
	\hline 
	$d\geq 3$ &beyond two-level   &  $\rho\in \bar{\mathfrak P}_H^{(N,1)}$   & $E(\rho;H) 
	\leq E_{{\beta}({\rho})}(H) \min^{\star}\left\{ \left( 1 - \tfrac{R(H)}{N} \right)^{-1},e^{\beta({\rho}) \epsilon_{\max} \frac{R(H)}{N}} \right\}$ \\ 
	& degenerate & $\rho\in {\mathfrak P}_H^{(N)}$, $\lambda^{\min}(0)  \geq Z_{\beta(\rho)}^{-1}$  &  $E(\rho;H) 
	\leq E_{{\beta}({\rho})}(H) e^{\beta({\rho}) \epsilon_{\max} \frac{R(H)}{N}} $ 
	\\ 
	&  &  $\rho\in {{\mathfrak P}_H^{(N)}}^*$, $\lambda^{\min}(0)  < Z_{\beta(\rho)}^{-1}$ &  $E(\rho;H)<  \frac{N}{N-2}\left[ 1 + \frac{u(\rho)R(H)}{N} \right] E_{\beta(\rho)}(H) +  \frac{ \beta^{-1}(\rho) +(d_0 - 1)Z_{\beta(\rho)}^{-1}\epsilon_{\max}}{N-2}  +  \mathcal{O} \left( \frac{1}{N^2} \right)$ 
	\\  \hline 
	\multirow{ 2}{*}{$d\geq 3$} &beyond two-level   &  $\rho\in {\mathfrak P}_H^{(N)}$   & \multirow{ 2}{*}{ $E(\rho; H) \leq \epsilon_{\max} (d-d_0) \exp\left[ -N \ln d_0 + (N-1) S(\rho)) \right]$ } \\
	& degenerate & $S(\rho) < \ln d_0$ & 
	\\  \hline \hline 
\end{tabular}
}
		\caption{Brief recapitulation of the relations between the mean energies $E(\rho;H)$ and $ E_{{\beta}({\rho})}(H)$ for $N$-passive  (possibly
	$1$-structurally stable) states $\rho$, and their  corresponding Gibbs isoentropic counterparts. 
	Following the notation introduced in Sec.~\ref{subsection1}, ${\mathfrak P}_H^{(N)}$ denotes the space of $N$-passive states while
	$\bar{\mathfrak P}_H^{(N,1)}$ the space of $N$-passive, $1$-structurally stable states; ${{\mathfrak P}_H^{(N)}}^*$ instead
	stands for the set of passive states with entropy larger than or equal to $\ln d_0$, where $d_0$ is the degeneracy of the
	zero-energy ground state of $H$. Two-level Hamiltonian 
	$H$ are those which, besides the zero-energy ground state level, have only another energy eigenvalue which is strictly positive.
	Finally  $\epsilon_{\max}$ is the maximum eigenvalue of $H$;  $R(H)$ is the spectral quantity defined in Eq.~(\ref{DEFR});
	$u(\rho)=\min\{1, \beta({\rho})\epsilon_{\max}\}$; 
	$\lambda^{\min}(0)$ is the minimum population value of $\rho$ corresponding to the zero-ground state energy (see Eq.~(\ref{DEFLAMBDAMIN})), while
	$Z_{\beta(\rho)}^{-1}$ is the associated population of the Gibbs counterpart. The superscript $\star$ on the $\min$ symbol
	means that the term $( 1 - \tfrac{R(H)}{N} )^{-1}$ only contribute for $R(H)<N$.
	We remind that for all $\rho$, irrespectively from 
	their passivity or non-passivity status, one always has  
	$E(\rho;H)\geq E_{{\beta}({\rho})}(H)$, whenever the isoentropic Gibbs counterpart of $\rho$ is definable -- a condition that applies for the cases treated in the table, except for the one in the last row, where we instead give an inequality as a function of the entropy $S(\rho)$.\label{TABLE1}} 
	\end{table*}

The manuscript is organized as follows: 
in Sec.~\ref{subsection1} we introduce the notation, set the theoretical framework
that will be used in the remaining part of the paper, and present some preliminary
observations.  Sec.~\ref{sec:inequality} contains the main result of the work: here we
derive our upper bound for the energy of $N$-passive, structurally stable states and 
discuss its achievability. Sec.~\ref{sec:struct} presents instead a generalization of the
bound for $N$-passive states which are not necessarily structurally stable, which applies
in the asymptotic limit of sufficiently large $N$.  
In Sec.~\ref{APPEC} we present finally some considerations on 
the case of Hamiltonian characterised by energy levels gaps which are commensurable. 
Conclusions are drawn 
in Sec.~\ref{sec:con}. The manuscript contains also few appendixes which provide technical support 
for the derivation of the main results (in particular in Appendix~\ref{PROFGI} we give a new proof of the fact that Gibbs states and
ground states are the only density matrices which are completely passive).

	\section{Definitions and preliminary observations}\label{subsection1} 

Let $A$ be a  quantum system described by a  Hilbert space ${\cal H}$ of finite dimension $d$ and 
characterised by an assigned Hamiltonian 
\begin{eqnarray}\label{HAM} 
H:=\sum_{j=0}^{d-1} \epsilon_j |\epsilon_j\rangle\langle \epsilon_j|\;,\end{eqnarray} 
with 
eigenvectors $\{|\epsilon_j\rangle\}_j$ and associated eigenvalues $\{ \epsilon_j\}_j$ which we assume to be organized in non-decreasing order, i.e. 
\begin{eqnarray}  \label{ORD1} 
\epsilon_{j+1}& \geq& \epsilon_j \;,\qquad  \forall j\in \{ 0,\cdots, d-2\}  \;.
\end{eqnarray}  
Notice that in the writing of (\ref{ORD1}) we are explicitly allowing for possible degeneracies in the spectrum of $H$. In particular for future reference 
we indicate with ${d_0}$ the degeneracy of its ground level (meaning that $\epsilon_i = \epsilon_0$ for all $ i \in\{ 0,1 ,\cdots,  {d_0}-1\}$), and use the symbol 
${\cal H}_G$ to represent the associated eigenspace, i.e.
\begin{eqnarray} 
{\cal H}_G: = \mbox{Span} \{  |\epsilon_0\rangle, |\epsilon_1\rangle, \cdots ,
|\epsilon_{{d_0}-1}\rangle\}\;.\label{SPAN} \end{eqnarray} 
Given then an element $\rho$ of the  density matrix set $\mathfrak{S}$ of $A$
we now define its ergotropy  as the functional   
\begin{eqnarray}  
{\cal E}^{(1)}(\rho;H) &:=& \max_U  \{ E(\rho; H) - E(U\rho U^\dag; H)\} \;, \label{ERGO}  
\label{EED1} 
\end{eqnarray}
where $E(\rho; H) : =  \mbox{Tr}[ \rho H]$ represents the mean energy of the state and 
where the maximization is performed over all possible unitary transformations $U$ acting on $A$~\cite{PASSIVE1,PASSIVE2,PASSIVE3,PASSIVE4}.  
The definition~(\ref{ERGO})   is explicitly invariant under rigid shifts  of the energy spectrum: accordingly without loss of generality, hereafter we shall restrict ourselves to the case of positive
semidefinite Hamiltonian $H$, with zero ground state energy value, i.e. 
\begin{eqnarray}
H \geq 0 \;, \qquad \epsilon_i =0 \quad \forall  i \in\{ 0, 1, \cdots {d_0}-1\} \;. \label{dde} 
\end{eqnarray} 
By construction ${\cal E}^{(1)}(\rho;H)$  is a non-negative quantity which can be explicitly 
computed by  solving the 
maximization with respect to $U$ 
(see Appendix~\ref{APPA} for details on this). 
In the above theoretical framework the set ${\mathfrak P}_H^{(1)}$ of passive states can now  be identified 
as the  subset  of  ${\mathfrak S}$  
characterised by the   
property 
of 
having  
zero ergotropy value, 
\begin{eqnarray} 
{\mathfrak P}_H^{(1)}  &: =& \Big\{ \rho\in {\mathfrak S} : {\cal E}^{(1)}(\rho; H) = 0 \Big\}\;. 
\end{eqnarray} 
It turns out that
every passive state $\rho$ must coincide with one of its  associated passive counterparts~\cite{PASSIVE1,PASSIVE2}
implying that 
${\mathfrak P}_H^{(1)}$ is exclusively made of density matrices which verify the following constraints: 
\begin{itemize}
	\item[i)] $\rho$ is diagonal in the energy eigenbasis $\{ |\epsilon_j\rangle\}_j$, i.e. it  can be expressed as 
	\begin{eqnarray} \label{RHOpas} \rho=\sum_{j=0}^{d-1} \lambda_j |\epsilon_j\rangle\langle \epsilon_j|\;;\end{eqnarray}
	\item[ii)] $\rho$ has no population inversion, i.e. its eigenvalues $\{ \lambda_j\}_j$ fulfil the following ordering \begin{eqnarray} 
	\epsilon_i > \epsilon_j \qquad \Longrightarrow \qquad \lambda_{i} \leq \lambda_{j} \;. \label{ORD2} 
	\end{eqnarray} 
\end{itemize} 
A special subset $\bar{\mathfrak P}_H^{(1,1)}$ of ${\mathfrak P}_H^{(1)}$ is provided by the passive states which are structurally stable, i.e. passive density matrices which besides obeying i) and ii) also satisfy the 
extra property 
\begin{eqnarray} 
\mbox{iii)} \qquad   \epsilon_i  =   \epsilon_{j} \qquad \Longrightarrow \qquad \lambda_{i} = \lambda_{j} \;, \qquad \qquad 
\qquad \label{EQUAL} 
\end{eqnarray} 
which assigns identical population values to energy eigenvalues which are degenerate~\cite{PASSIVE1,PASSIVE2}. One can easily verify that  when $H$ has a not-degenerate spectrum (i.e. when (\ref{ORD1}) holds true with strict inequalities
for all $j$), 
$\bar{\mathfrak P}_H^{(1,1)}$ 
coincides with ${\mathfrak P}_H^{(1)}$, while otherwise it constitutes a proper subset of the latter. 
Finally both ${\mathfrak P}_H^{(1)}$ and   $\bar{\mathfrak P}_H^{(1,1)}$ can be shown to be closed under convex convolution.

In a similar fashion, given $N\geq 1$ integer, we can now introduce the set ${\mathfrak P}_H^{(N)}$ of $N$-ordered passive (or simply $N$-passive)
states. Specifically,   given 
\begin{eqnarray} H^{(N)} : =\sum_{\ell=1}^N  H^{(\ell)}, \end{eqnarray} 
the total Hamiltonian of the joint system,
$H^{(\ell)}$ being the single system Hamiltonian acting on the $\ell$-th copy 
we define 
\begin{eqnarray} 
{\mathfrak P}_H^{(N)}  &: =& \Big\{ \rho\in {\mathfrak S} : {\cal E}^{(N)}(\rho; H) = 0 \Big\}\;,
\end{eqnarray} 
with ${\cal E}^{(N)}(\rho; H)$ the $N$-order ergotropy  functional 
\begin{equation} \label{NNN} 
{\cal E}^{(N)}(\rho; H):= 
\max_{U}  \{ E(\rho^{\otimes N}; H^{(N)}) - E(U\rho^{\otimes N} U^\dag; H^{(N)})\}
\;, \end{equation}  
the maximum being now evaluated with respect to  all the (possibly non-local) unitaries of the joint system.
Using the same argument that
led to the necessary and sufficient conditions i) and ii) given above, 
it has been shown that  $\rho\in {\mathfrak P}_H^{(N)}$ 
if and only if the following conditions hold true:  
\begin{itemize}
	\item[\it i)]  $\rho$ is diagonal in the energy eigenbasis, i.e. must be of the form~(\ref{RHOpas});
	\item[ \it ii)] its eigenvalues fulfil  the requirement
	\begin{equation}
	\label{eq1}
	\sum_{i=0}^{d-1}  n_i \epsilon_i >   \sum_{j=0}^{d-1}  m_j  \epsilon_j \quad \Longrightarrow \quad  
	\prod_{i=0}^{d-1}   \lambda_i^{n_i}  \leq    \prod_{j=0}^{d-1}  \lambda_j^{m_j} , 
	\end{equation}
	for all couples of the population sets $I_N:= \{ n_0, n_1, \cdots, n_{d-1}\}$, $J_N:=  \{ m_0,m_1, \cdots, m_{d-1}\}$  formed by $d$ non-negative integers
	that sum up to $N$. 
\end{itemize} 
Notice that for $N=1$ Eq.~(\ref{eq1}) reduces to
the absence of population inversion (\ref{ORD2}), and that for consistency 
in the above expression the indeterminate form $0^{0}$ has to be interpreted equal to~$1$, i.e.
explicitly 
\begin{eqnarray} 
(\lambda=0)^{(n=0)} = 1 \;. \label{RULE1} 
\end{eqnarray} 
By close inspection of the above definitions, it follows that   $N$-passive
states are also $N'$-passive for all $N'\leq N$. The opposite inclusion however is not necessarily
granted  implying a specific  ordering 
on the associated sets, i.e. 
\begin{eqnarray} {\mathfrak P}_H^{(N)}  \subseteq {\mathfrak P}_H^{(N')}
\subseteq {\mathfrak P}_H^{(1)} \;, \quad \forall N' \leq N. \label{ORDE1} 
\end{eqnarray} 

In a similar fashion of what done in the case of ${\mathfrak P}_H^{(1)}$, also for ${\mathfrak P}_H^{(N)}$
we can then introduce the notion structural stability. Specifically, for any given integer $k\leq N$ we define 
the subset $\bar{\mathfrak P}_H^{(N,k)}$ of the $N$-passive, $k$-structurally stable density matrices 
as the one formed by the special $N$-passive elements  
which, besides the condition {\it i)} and {\it ii)} detailed above  also obey to the extra requirement
\begin{equation}
\label{eq11}
\mbox{\it iii)} \quad \sum_{i=0}^{d-1}  n_i \epsilon_i =   \sum_{j=0}^{d-1}  m_j  \epsilon_j \quad \Longrightarrow \quad  
\prod_{i=0}^{d-1}   \lambda_i^{n_i}  =    \prod_{j=0}^{d-1}  \lambda_j^{m_j}\;,
\end{equation}
\begin{itemize}
	\item[] for all the couples of the population sets $I_k:= \{ n_0, n_1, \cdots, n_{d-1}\}$, $J_k:=  \{ m_0,m_1, \cdots, m_{d-1}\}$  formed by $d$ non-negative integers
	that sum up to $k$;
\end{itemize}  
(again in writing~(\ref{eq11})  
we assume the convention (\ref{RULE1}); notice also that for $k=N=1$ 
this expression reduces to (\ref{EQUAL})). 

Hierarchical rules analogous 
to (\ref{ORDE1}) 
hold true also for the sets $\bar{\mathfrak  P}^{(N,k)}$. In this case,  one
can easily show that for every $N\geq k\geq k'\geq 1$
\begin{equation} \label{ORDE20} \bar{\mathfrak P}_H^{(N,k)}\subseteq \bar{\mathfrak P}_H^{(N,k')}\subseteq \bar{\mathfrak P}_H^{(N,1)}  \subseteq {\mathfrak P}_H^{(N)}
\;,\end{equation} 
and also that 
\begin{eqnarray} \label{ORDE21} \bar{\mathfrak P}_H^{(N,k')}\subseteq \bar{\mathfrak P}_H^{(N',k')}\;,  \quad \forall N\geq N'\geq k'\geq 1\;,\end{eqnarray} 
while in general for  $N\geq N'\geq 1$ it is not true that    
${\mathfrak P}_H^{(N)}$ is contained in $\bar{\mathfrak P}_H^{(N',1)}$.
In the special case  of $1$-structurally stable configurations (i.e. $k'=1$), Eq.~(\ref{ORDE21}) implies 
\begin{eqnarray} \label{ORDE2} \bar{\mathfrak P}_H^{(N,1)}\subseteq \bar{\mathfrak P}_H^{(N',1)}
\subseteq \bar{\mathfrak P}_H^{(1,1)}\;,  \quad \forall N\geq N'\geq 1\;,\end{eqnarray} 
which, for all $N\geq 1$, allows us to express $\bar{\mathfrak P}_H^{(N,1)}$ as the inclusion between 
${\mathfrak P}_H^{(N)}$ and $\bar{\mathfrak P}_H^{(1,1)}$, i.e. 
\begin{eqnarray} \label{EXPINCLU} \bar{\mathfrak P}_H^{(N,1)}= {\mathfrak P}_H^{(N)}
\bigcap \bar{\mathfrak P}_H^{(1,1)}\;, \end{eqnarray} 
(to verify this simply observe that on one hand, $\bar{\mathfrak P}_H^{(N,1)}$ is certainly contained in
${\mathfrak P}_H^{(N)}
\bigcap \bar{\mathfrak P}_H^{(1,1)}$ because it is a subset of both ${\mathfrak P}_H^{(N)}$
and $\bar{\mathfrak P}_H^{(1,1)}$. On the other hand  by definition all the elements of ${\mathfrak P}_H^{(N)}
\bigcap \bar{\mathfrak P}_H^{(1,1)}$ are $N$-passive and $1$-structurally stable, hence
elements of $\bar{\mathfrak P}_H^{(N,1)}$). 
In the special case of $H$ which is non-degenerate (i.e. when (\ref{ORD1}) holds true with strict inequalities
for all $j$) we have already noticed that ${\mathfrak P}_H^{(1)}=  \bar{\mathfrak P}_H^{(1,1)}$, accordingly exploiting (\ref{ORDE1}) it follows that 
Eq.~(\ref{EXPINCLU}) leads to the conclusion that
\begin{equation} \label{EXPINCLUHnondeg} 
\mbox{($H$ = non-deg)} \quad \Longrightarrow \quad \bar{\mathfrak P}_H^{(N,1)}= {\mathfrak P}_H^{(N)}\;. 
\end{equation}

We are finally in the position to give the  definition of {\it completely passive} (CP) states: 
these are the density matrices of $A$  which are passive at all order $N$, i.e. which are diagonal 
in the energy eigen-basis and fulfil the constraint~(\ref{eq1}) at all order. The set of completely passive can be 
hence identified with the intersection of all the ${\mathfrak{P}_H^{(N)}}$, i.e. 
\begin{eqnarray} 
\bigcap_{N \ge 1} \label{CPdef} 
{\mathfrak{P}_H^{(N)}} = {\mathfrak{P}_H^{(\infty)}} \;,\end{eqnarray} 
the last identity being a trivial consequence of the ordering 
(\ref{ORDE1}).
In a similar fashion we can also introduce the definition of  CP states which are also
1-structurally stable (CP1SS), as the intersection of all the sets ${\mathfrak{P}_H^{(N,1)}}$, i.e. 
\begin{eqnarray} \label{CP1SSdef}
\bigcap_{N \ge 1} 
\bar{\mathfrak{P}}_H^{(N,1)} = \bar{\mathfrak{P}}_H^{(\infty,1)} \;,\end{eqnarray} 
as well as the set of CP states which are
structurally stable at all orders (CPCSS) which, thanks to  (\ref{ORDE20}) corresponds
to the
inclusion of all the sets $\bar{\mathfrak P}_H^{(N,N)}$, i.e. 
\begin{eqnarray} 
\bigcap_{N \ge 1} 
\bar{\mathfrak{P}}_H^{(N,N)} = \bar{\mathfrak{P}}_H^{(\infty,\infty)} \;.\end{eqnarray} 
By construction it is clear that
$\bar{\mathfrak{P}}_H^{(\infty,\infty)} $ is included in $\bar{\mathfrak{P}}_H^{(\infty,1)}$
which in turns is a subset of ${\mathfrak{P}}_H^{(\infty)}$. 	Most notably however it turn out that irrespectively from $H$, 
$\bar{\mathfrak{P}}_H^{(\infty,\infty)}$ coincides with $\bar{\mathfrak{P}}_H^{(\infty,1)}$, 
\begin{eqnarray} \label{NOTABLE} 
\bar{\mathfrak{P}}_H^{(\infty,\infty)} = \bar{\mathfrak{P}}_H^{(\infty,1)} \;,\end{eqnarray} 
both being
identified  with the set  of Gibbs states of the system. This important  result will be reviewed in the next
subsection, with a complete characterisation of ${\mathfrak{P}}_H^{(\infty)}$ in 
terms of  ground and Gibbs states. 

\subsection{Ground states and Gibbs states} \label{SUB1} 

Explicit examples of CP states are provided by the ground states 
density matrix $\rho$ which have their support inside the ground subspace ${\cal H}_G$~\cite{PASSIVE2}.
We shall use the symbol
${\mathfrak S}_H^{(G)}$ to indicate the associated subset, i.e.  
\begin{eqnarray}  
{\mathfrak S}_H^{(G)} := \{ \rho\in{\mathfrak S}: \rho \Pi_G = \Pi_G \rho=\rho\}\;,
\end{eqnarray} 
with $\Pi_G= \sum_{j=0}^{{d_0}-1} |e_i\rangle\langle e_j|$ being the projector on  ${\cal H}_G$. 
Notice also that, while ${\mathfrak S}_H^{(G)}$ is included into ${\mathfrak{P}}_H^{(\infty)}$, for ${d_0}>1$
its only element that is structurally stable is the uniform ground state mixture $\Pi_G/{d_0}$.
As a matter fact, in Ref.~\cite{PASSIVE2} it has been shown that
these  elements of ${\mathfrak S}_H^{(G)}$  are the only examples of CP states which are not CP1SS, i.e. 
\begin{eqnarray} 
\rho \in 
{{\mathfrak{P}}_H^{(\infty)}}  / {\bar{\mathfrak{P}}_H^{(\infty,1)}}  \quad \Longrightarrow \quad
\rho \in {\mathfrak S}_H^{(G)}  \;,  \label{GROUND} 
\end{eqnarray} 
or more explicitly 
\begin{eqnarray} 
{{\mathfrak{P}}_H^{(\infty)}}  / {\bar{\mathfrak{P}}_H^{(\infty,1)}}  = 
{\mathfrak S}_H^{(G)} / \left\{ \tfrac{\Pi_G}{{d_0}}\right\}   \;.  \label{GROUND0} 
\end{eqnarray} 

Closely related to ${\mathfrak S}_H^{(G)}$ is the 
set of Gibbs thermal states ${\mathfrak G}_H$.
They are  identified with the collection of density matrices
of the form 
\begin{eqnarray}\label{GIBBS} 
\omega_\beta := e^{-\beta H}/Z_\beta \;, \qquad Z_\beta: = \mbox{Tr}[ e^{-\beta H}] \;, 
\end{eqnarray}  with the parameter $\beta\geq 0$ playing the role of an effective inverse temperature, and the normalization term $Z_\beta$ being the associated partition function.  
In the infinite temperature regime, $\omega_\beta(H)$ converges to the completely mixed state of $A$, that is $\omega_0 : = \lim_{\beta \rightarrow 0} \omega_\beta = \mathbb{1} /d$. On the contrary, 
in the zero temperature limit, 
$\omega_\beta(H)$  converges to the uniform mixture supported on the ground energy eigenspace (\ref{SPAN}) of the system, namely
\begin{eqnarray} 
\omega_\infty &: =& \lim_{\beta \rightarrow \infty} \omega_\beta = {\Pi}_G/{{d_0}}  \;,  \label{LIMinf} 
\end{eqnarray} 
which, as we already mentioned, is the special element of  ${\mathfrak S}_H^{(G)}$.
Gibbs states   play a fundamental role in the study of the thermodynamic properties of $A$  since they embed the very notion
of thermal equilibrium~\cite{ALICKIREV,REVQTHERMO}.  
In particular they can be identified   by the property of
granting the minimal  value of the mean energy~$E(\rho; H)$  attainable for density matrices $\rho$ with fixed von Neumann entropy $S(\rho) := - \mbox{Tr} [ \rho \ln \rho]$, i.e. \begin{eqnarray} \label{MINEN} 
\min_{\{ \rho \in  \mathfrak{S}: S(\rho) = S_\beta\}}   E(\rho; H)= E_\beta(H) \;, 
\end{eqnarray}
with 
\begin{eqnarray}  \label{DEFNENBETA0} 
E_\beta(H) &:=& E(\omega_\beta; H) = 
-  \frac{\partial}{\partial \beta} \ln Z_\beta \;, \\  
S_\beta&: =& S(\omega_\beta) = \beta E(\omega_\beta; H) + \ln Z_\beta \;. \label{DEFENTROBETA} 
\end{eqnarray} 
It is worth remembering that both $E_\beta(H)$ and $S_\beta$ are strictly monotonically decreasing functions  of $\beta$. Thanks to this fact, there is
a one-to-one correspondence between these functionals, that ultimately leads to the following relevant expression 
\begin{eqnarray} \label{RELEVANT} 
\frac{\partial E_\beta(H)}{\partial S_\beta} = \frac{1}{\beta}\;. 
\end{eqnarray} 
We also notice  that, while $S_\beta$ can saturate the maximum entropy value available for the system $A$ (i.e. $S_{\beta=0}=\ln d$),
due to Eq.~(\ref{LIMinf}), the Gibbs state entropy functional is bound to be always larger than or equal to $\ln {d_0}$
which,
unless the Hamiltonian has a non-degenerate ground level (i.e. ${d_0}=1$),
is strictly larger than zero. Accordingly, (\ref{MINEN}) identifies the Gibbs states as minimal energy states only 
for the subclass of density matrices $\rho$ which have entropy above the $\ln {d_0}$. For those instead
which have $S(\rho) < \ln {d_0}$, the only possible lower bound on $E(\rho;H)$ is simply provided by elements of 
${\mathfrak S}_H^{(G)}$, leading to the following trivial inequality 
\begin{eqnarray} \label{MINENnew} 
\min_{\{ \rho \in  \mathfrak{S}: S(\rho) = S\}}   E(\rho; H) =0 \;,  \qquad \forall S < \ln d_0\;. 
\end{eqnarray}

As mentioned in the introductory section, an alternative way of characterising  
the Gibbs subset ${\mathfrak G}_H$ is by the observation that
the density matrices of the form~(\ref{GIBBS}) 
share the exclusive  property of being the only density matrices of $A$ which are CP and 1-structurally stable~\cite{PASSIVE1,PASSIVE2}. Furthermore, since the elements of ${\mathfrak G}_H$ 
are also structurally stable at all order, we can write 
\begin{eqnarray} \label{IDE1} 
{\mathfrak G}_H={\bar{\mathfrak{P}}_H^{(\infty,1)}}={\bar{\mathfrak{P}}_H^{(\infty,
		\infty)}} \;,  
\end{eqnarray} 
which explicitly proves the identity (\ref{NOTABLE})
anticipated at the end of the previous section. Together with~(\ref{GROUND0}), the above expression finally allows us
to conclude that 	\begin{eqnarray}  \label{GREATRES} 
{{\mathfrak{P}}_H^{(\infty)}}  = {\mathfrak G}_H \bigcup {\mathfrak S}_H^{(G)}\;.
\end{eqnarray} 
In Appendix~\ref{PROFGI} we present an alternative proof of these important identities based on a simple geometrical argument. 
Here instead we comment on the fact that the they can be simplified in two limiting cases.
The first one is  
for two dimensional systems ($d=2$). In this situation all passive states  are also structurally stable and completely passive, hence thermal, implying that the hierarchies (\ref{ORDE20}) collapse, 
so that 
Eqs.~(\ref{IDE1}) and (\ref{GREATRES}) can be replaced by 
\begin{eqnarray} (d=2)\quad \Longrightarrow  \quad     {\mathfrak{P}}_H^{(1)} =  \bar{\mathfrak{P}}_H^{(1,1)} = {\mathfrak G}_H \;, \label{easy1} 
\end{eqnarray} 
the ground state set being trivially contained in the Gibbs set, i.e. ${\mathfrak S}_H^{(G)} \subset {\mathfrak G}_H$.
A similar statement can be extended also for  $d\geq 3$ for Hamiltonians  that have a two-level spectrum, i.e. $H$ which beside the (possibly degenerate) 
zero, ground energy level are
characterised by a unique (possibly degenerate) non-zero eigenvalue. 
Indeed, in this case one can easily show that  $\bar{\mathfrak P}_H^{(1,1)}$ coincides with the Gibbs set, allowing us
to replace (\ref{IDE1}) with 
\begin{eqnarray} \left( d \geq 3, H=\mbox{two-level} \right) \quad \Longrightarrow  \quad 
\bar{\mathfrak{P}}_H^{(1,1)} = {\mathfrak G}_H \;, \label{easy2} 
\end{eqnarray} 
while (\ref{GREATRES}) remains the same.

\section{Upper bounds for the mean energy of $N$-passive, $1$-structurally stable states}
\label{sec:inequality}

Equation~(\ref{MINEN}) establishes that Gibbs states provide a natural lower bound for the mean energy value of the density matrices which have the same entropy. In particular, given a $N$-passive state $\rho \in {\mathfrak P}_H^{(N)}$
of entropy larger than or equal to $\ln {d_0}$, identifying 
the inverse temperature $\beta(\rho)\in [0, \infty]$  such that  $\omega_{\beta(\rho)}\in{\mathfrak G}_H$
has entropy $S_{\beta(\rho)}$ equal to $S(\rho)$, i.e. 
\begin{eqnarray}
S_{\beta(\rho)} &=& S(\rho) \;, \label{sameE}  
\end{eqnarray}
we can write
\begin{eqnarray} \label{EmaggiorediE0} 
E(\rho; H) &\geq& E_{\beta(\rho)}(H)\;.
\end{eqnarray}
Notice that, thanks to the uniform population constraint (\ref{EQUAL}), 
the requirement of having entropy larger than or equal to $\ln {d_0}$ is naturally fulfilled
by the elements of ${\mathfrak P}_H^{(N)}$ which are structurally stable  (i.e.
$S(\rho) \geq  \ln {d_0}\;,  \forall \rho\in \bar{\mathfrak P}_H^{(1,1)}$), accordingly for such 
states we need not to worry about such condition. 

For two dimensional systems ($d=2$), thanks to (\ref{easy1}) 
the inequality in~(\ref{EmaggiorediE0}) is replaced by an identity. 
By exploiting (\ref{easy2})  a similar statement can be extended  also for  $d\geq 3$ for those $H$ which have a  simple 
spectrum: 
in this case however the equality in (\ref{EmaggiorediE0}) 
holds true only for the passive states which are also explicitly structurally stable, i.e.
$\rho\in \bar{\mathfrak P}_H^{(1,1)}$, i.e. 
\begin{eqnarray}
E(\rho; H) = E_{\beta(\rho)}(H) \quad   \begin{cases}
\forall \rho\in {\mathfrak P}_H^{(1)}, &
(d=2);\\  \\ \forall \rho\in \bar{\mathfrak P}_H^{(1,1)}, 
& (d\geq 3, H=\mbox{two-level});   \end{cases}  .  \label{easy3} 
\end{eqnarray} 
On the contrary, for  $d\geq 3$ and Hamiltonian  $H$  that possesses at least two distinct non-zero energy eigenvalues,  
the gap between the right-hand-side and left-hand-side of Eq.~(\ref{EmaggiorediE0}) is typically finite. Yet, as a consequence
of  Eq.~(\ref{IDE1}), we expect that such gap 
should reduce as  $N$ increases, irrespectively from the spectral properties of $H$. Aim of the present section is to provide a quantitative estimation of this fact at least for the elements of ${\mathfrak P}_H^{(N)}$ which are at least $1$-structurally stable. In particular we shall see that  the following upper bound hold true, 
\begin{equation} \label{RESULTexp} 
E(\rho;H)  \leq  E_{\beta(\rho)}(H) e^{\beta(\rho) \epsilon_{\max} \frac{R(H)}{N}} \;,   \quad \forall \rho\in \bar{\mathfrak P}_H^{(N,1)}\;,
\end{equation} 
with $\epsilon_{\max} (=\epsilon_d)$ being the maximum energy eigenvalue of $H$ and with 
$R(H)$ being a non-negative constant  that only depends upon the spectral properties of the system Hamiltonian. Explicitly
for $H$ not having a simple (two-level) spectrum one has
\begin{eqnarray}
R(H) := \underset{\epsilon_c > \epsilon_b > \epsilon_a}\max \; \frac{\epsilon_c- \epsilon_a}{\epsilon_b - \epsilon_a}, \label{DEFR} \end{eqnarray}
the maximum being computed among all possible triples of ordered energy levels $\epsilon_c > \epsilon_b > \epsilon_a$ -- for $H$ being two-level we can just put $R(H)=0$ and recover (\ref{easy3}) via (\ref{EmaggiorediE0}).
Notice that 
for $H$ with non-simple spectrum $R(H)$ is always greater than or equal to one and that
the optimization over $\epsilon_c$ can be explicitly carried out leading to 
\begin{equation}
R(H) = \underset{\epsilon_b > \epsilon_a}\max \; \frac{\epsilon_{\max}- \epsilon_a}{\epsilon_b - \epsilon_a} =
1 + \underset{\epsilon_b > \epsilon_a}\max \; \frac{\epsilon_{\max}- \epsilon_b}{\epsilon_b - \epsilon_a}\label{DEFR1} \geq 1\;.\end{equation}
Interestingly enough, when $\beta(\rho) \epsilon_{\max} > 1$ (low temperature regime), 
the upper bound of Eq.~(\ref{RESULTexp}) can be improved by means of the following 
inequality 
\begin{equation} \label{RESULT} 
E(\rho;H)  \left( 1 - \tfrac{R(H)}{N} \right) \leq  E_{\beta(\rho)}(H)  \;,   \quad \forall \rho\in \bar{\mathfrak P}_H^{(N,1)}
\end{equation} 
which, while being trivial for $R(H)/N\geq 1$,  for 
$R(H)/N < 1$ leads to 
\begin{equation} \label{RESULTinv} 
E(\rho;H)  \leq  E_{\beta(\rho)}(H) \left( 1 - \tfrac{R(H)}{N} \right)^{-1} \;,  \quad \forall \rho\in \bar{\mathfrak P}_H^{(N,1)}\;.
\end{equation} 
A direct comparison between (\ref{RESULTexp}) and (\ref{RESULTinv}) reveals that when $\beta(\rho) \epsilon_{\max} >1$ and $N$ is sufficiently large, 
the former is tighter than the latter (otherwise  (\ref{RESULTexp}) always wins). 

\subsection{Derivation of the bounds} \label{sec:der} 
In order to derive Eqs.~(\ref{RESULTexp}) and (\ref{RESULT}) 
first of all, we need to understand in which way the eigenvalues of the state $\rho$ are constrained by each other by the requirement of being an element
of the $N$-passive set ${\mathfrak P}_H^{(N)}$. 
Since the problem is non trivial only for the case where the spectrum of  $H$ is not two-level, in what follows we shall  
focus on these cases  where $H$ admits at least 3 different energy levels. 
\begin{prep}\label{PREP1} 
	Consider an Hamiltonian $H$ admitting three non-degenerate energy levels 
	$\epsilon_a<\epsilon_b<\epsilon_c$, and
	$\rho \in {\mathfrak P}_H^{(N)}$  a $N$-passive quantum state with associated populations 
	$\lambda_a\geq \lambda_b\geq \lambda_c$. 
	Then for  $m\in [0 ,N-1]$  integer such that  
	\begin{eqnarray} \label{CONDN} \frac{m}{N} < \frac{\epsilon_b - \epsilon_a}{\epsilon_c - \epsilon_a}  \leq  \frac{m+1}{N} 
	\;, \end{eqnarray} 
	the following inequality  holds
	\begin{eqnarray}  \lambda_b \leq \lambda_c^{\frac{m}{N}}\lambda_a^{1-\frac{m}{N}} \label{primaprova1}  
	\;. \end{eqnarray} 
	Viceversa 
	for  $m\in [0 ,N-1]$  integer such that  
	\begin{eqnarray} \label{CONDN1}
	\frac{m}{N} \leq  
	\frac{\epsilon_b - \epsilon_a}{\epsilon_c- \epsilon_a}  < \frac{m+1}{N} 
	\;, \end{eqnarray} 
	then  we must have 
	\begin{equation} \label{primaprova}  
	\lambda_b \geq \lambda_c^{\frac{m+1}{N}}\lambda_a^{1-\frac{m+1}{N}}  \;.\end{equation} 
\end{prep}
\begin{proof} 
	The result follows directly from the $N$-passivity condition Eq.~(\ref{eq1}) and results
	in the geometrical constraint depicted in Fig.~\ref{fig:Prep1}.
	In particular  Eq.~(\ref{primaprova1})
	corresponds to imposing~(\ref{eq1})
	for the special case in which the population sets $I_N$ contains as only non-zero term $n_b=N$, while $J_N$ contains as only non-zero terms 
	$n_c=m$ and $n_a=N-m$. 
	The  inequality~(\ref{primaprova}) instead follows by applying  Eq.~(\ref{eq1}) 
	to the special case in which $I_N$ contains as only non-zero terms $n_c=m+1$, $n_a=N-m-1$, while $J_N$ contains as only non-zero term $n_b=N$. 
\end{proof}

\begin{figure}
	\centering
	\includegraphics[width=250px]{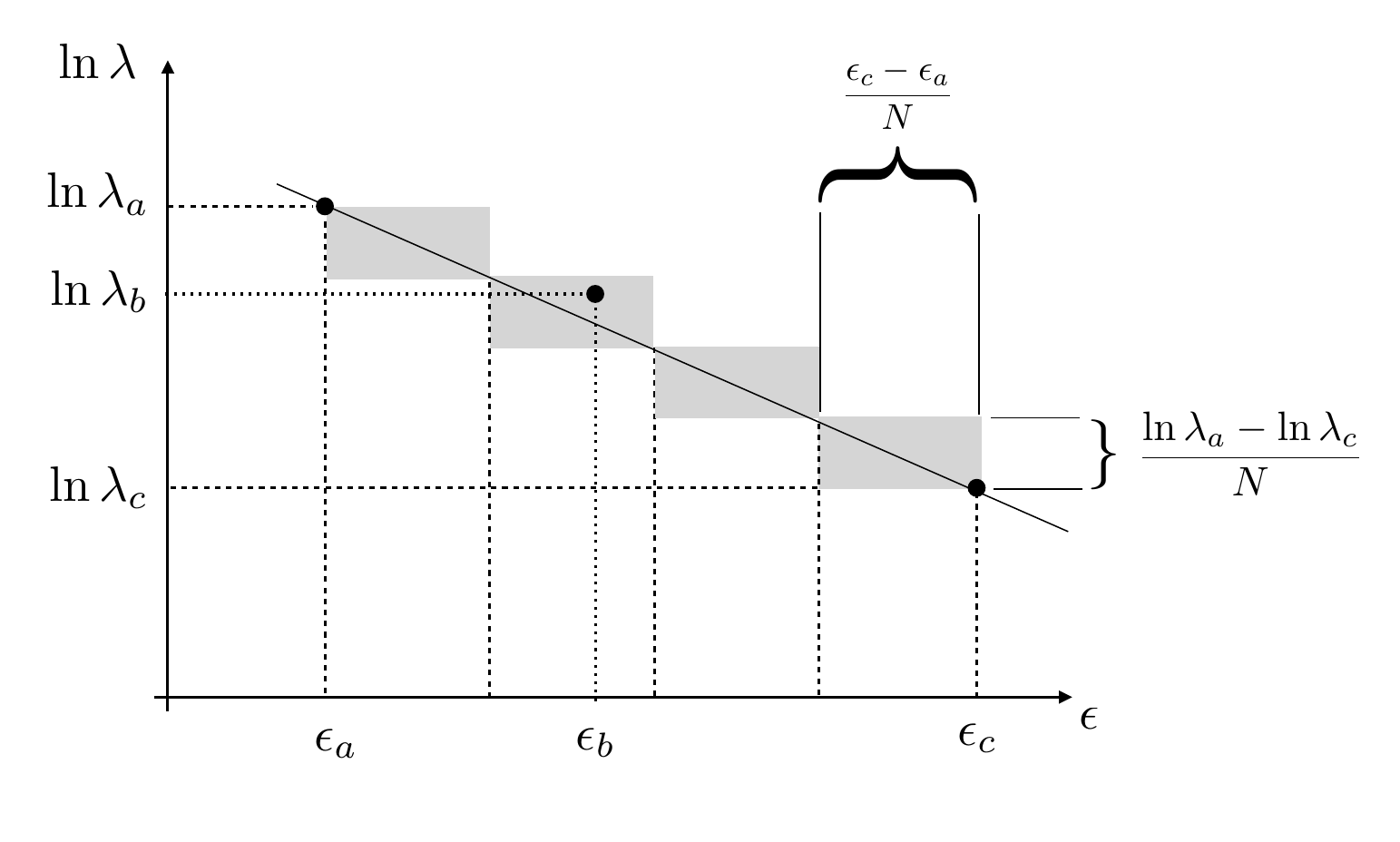}
	\caption{Graphical visualisation of {\bf Proposition \ref{PREP1} } in the $(\epsilon, \ln \lambda)$ plane for a $N$-passive state with $N=4$. The shaded area represents the allowed region for the value of $\ln \lambda_b$ associated with an intermediate energy level  $\epsilon_b$ which follows from Eqs.~(\ref{primaprova1}) and~(\ref{primaprova}). The continuous line in the plot corresponds is the linear interpolation between the extremal points $(\epsilon_a, \ln \lambda_a)$, $(\epsilon_c, \ln \lambda_c)$.
		Identifying $i=b$, $j=c$, $k=a$, condition~(\ref{CONDN1}) implies that the continous line lie below, and has a smaller (more negative) angular coefficient, than the line defined by the equation $\ln \lambda = -\ln Z - \beta\epsilon$ (not depicted).	The inequality~(\ref{implic1}) then follows by $\frac{\ln \lambda_a - \ln \lambda_c}{N} < \beta \frac{\epsilon_j - \epsilon_k}{N}$.  }
	\label{fig:Prep1}
\end{figure}
Notice that, in \textbf{Proposition~\ref{PREP1}}, the upper bound of condition~(\ref{CONDN}) is not necessary in order to prove~(\ref{primaprova1}), and likewise the lower bound of~(\ref{CONDN1}) is not necessary to prove~(\ref{primaprova}). However, they will be used in the next proposition, in which we will proceed to estimate how much the populations of a $N$-passive, $1$-structurally stable state $\rho$ can differ from those of a thermal equilibrium matrix. 
\begin{prep}
	\label{prep2} Let $\epsilon_k < \epsilon_j$ be two distinct  eigenvalues of $H$ and 
	$\rho \in {\mathfrak P}_H^{(N)}$   a $N$-passive state.
	Then if  there exist
	real numbers $\beta, \, {Z} \in [0, \infty)$ fulfilling the inequalities 
	\begin{eqnarray} 
	\lambda_k \leq Z^{-1}\;, \qquad \lambda_j \leq {Z}^{-1}
	e^{-\beta (\epsilon_j -\epsilon_k)} \;, \label{CONDPREP2} 
	\end{eqnarray}  the following implication must hold 
	\begin{equation} \label{implic1_inst} 
	\epsilon_k\ <  \epsilon_i <  \epsilon_j \implies \ln \lambda_i \leq  -\beta (\epsilon_i -\epsilon_k)- \ln Z +  \tfrac{\beta ( \epsilon_j -\epsilon_k)}{N}
	.\end{equation}
	Instead if 
	there exist
	real numbers $\beta, \, {Z} \in [0, \infty)$ fulfilling at least one of the two following inequalities:
	\begin{eqnarray} 
	\lambda_k \leq Z^{-1}\;, \qquad \lambda_j \leq \lambda_k 
	e^{-\beta (\epsilon_j -\epsilon_k)} \;, \label{CONDPREP21} 
	\end{eqnarray}
	or
	\begin{eqnarray} 
	\lambda_k \geq Z^{-1}\;, \qquad \lambda_j \leq Z^{-1}
	e^{-\beta (\epsilon_j -\epsilon_k)} \;, \label{CONDPREP21_magmin} 
	\end{eqnarray}
	then we must have 
	\begin{equation}   \label{implic2_inst} 
	\epsilon_k < \epsilon_j <  \epsilon_i \Longrightarrow 
	\ln \lambda_i \leq  -\beta(\epsilon_i -\epsilon_k)- \ln {Z} +  \beta\tfrac{(\epsilon_i - \epsilon_k)^2}{N(\epsilon_j - \epsilon_k)} \;. 
	\end{equation} 
\end{prep}
\begin{proof}
	To show Eq.~(\ref{implic1_inst}) 
	let us set  $\epsilon_a=\epsilon_k$, $\epsilon_b=\epsilon_i$, and $\epsilon_c=\epsilon_j$ and select $m\in[0,N - 1]$ such that
	Eq.~(\ref{CONDN}) of  {\bf Proposition \ref{PREP1}} holds true. Then from (\ref{primaprova1}) we get 
	\begin{eqnarray}
	\lambda_i &\leq& \lambda_j^{\frac{m}{N}}\lambda_k^{1 - \frac{m}{N}} \leq  Z^{-1}\; e^{-\beta (\epsilon_j -\epsilon_k) \frac{m}{N}}  \nonumber \\   
	&\leq& Z^{-1}e^{-\beta  (\epsilon_j -\epsilon_k)   \frac{N{ (\epsilon_i-\epsilon_k) } -  (\epsilon_j -\epsilon_k) }{N  (\epsilon_j -\epsilon_k) } } 
	\nonumber \\
	&=& Z^{-1} e^{-\beta  (\epsilon_i -\epsilon_k) }  e^{\frac{\beta (\epsilon_j -\epsilon_k)  }{N}}\;.
	\end{eqnarray}\
	To prove instead Eq.~(\ref{implic2_inst}) set $\epsilon_a=\epsilon_k$, $\epsilon_b=\epsilon_j$, and $\epsilon_c=\epsilon_i$ and select $m\in[0,N-1]$ such that 
	Eq.~(\ref{CONDN1}) of {\bf Proposition  \ref{PREP1}} applies. Then we have 
	\begin{equation}
	\label{firststep}
	\lambda_i \leq (\lambda_j)^{\frac{N}{m+1}}(\lambda_k)^{1 - \frac{N}{m+1}}\;. 
	\end{equation}
	In the case in which inequalities (\ref{CONDPREP21_magmin}) hold, exploiting the fact that
	$1-N/(m+1) \leq 0$ 
	we can bound the RHS of (\ref{firststep}) by
	\begin{multline}
	\lambda_i \leq (\lambda_j)^{\frac{N}{m+1}}(\lambda_k)^{1 - \frac{N}{m+1}} \leq  \left(Z^{-1}
	e^{-\beta (\epsilon_j -\epsilon_k)} \right)^{\frac{N}{m+1}}(Z^{-1})^{1 - \frac{N}{m+1}} = 
	{Z}^{-1} (e^{-\beta(\epsilon_j - \epsilon_k)})^{\frac{N}{m+1}}\;. \\   
	\end{multline}
	Also, if the inequalities (\ref{CONDPREP21}) are true, we can write the same bound for $\lambda_i$:
	\begin{equation}
	\lambda_i \leq (\lambda_j)^{\frac{N}{m+1}}(\lambda_k)^{1 - \frac{N}{m+1}} \leq 
	\lambda_k (e^{-\beta(\epsilon_j - \epsilon_k)})^{\frac{N}{m+1}}  
	\leq {Z}^{-1} (e^{-\beta(\epsilon_j - \epsilon_k)})^{\frac{N}{m+1}}\;. 
	\end{equation}
	Therefore, in either one of the two cases (\ref{CONDPREP21}) and (\ref{CONDPREP21_magmin})
	we have
	\begin{multline}
	\lambda_i \leq {Z}^{-1} (e^{-\beta(\epsilon_j - \epsilon_k)})^{\frac{N}{m+1}} 
	\leq 
	{Z}^{-1}  \exp{-\beta(\epsilon_i - \epsilon_k)\left(1 + \tfrac{(\epsilon_i - \epsilon_k)}{N(\epsilon_j - \epsilon_k)} \right)^{-1} } \\
	\leq  {Z}^{-1}  \exp{ -\beta(\epsilon_i - \epsilon_k)\left( 1 - \tfrac{(\epsilon_i - \epsilon_k)}{N(\epsilon_j - \epsilon_k)} \right) } 
	= {Z}^{-1} e^{-\beta(\epsilon_i-\epsilon_k)} \exp{\beta\tfrac{(\epsilon_i - \epsilon_k)^2}{N(\epsilon_j - \epsilon_k)} }\;, 
	\end{multline}
	where  in the last inequality we used  $1-x\leq {1}/({1+x})$.
\end{proof}


The results of {\bf Proposition \ref{prep2}} apply to 
density matrices $\rho$ which are 
just $N$-passive, but not necessarily $1$-structurally stable. In order to treat the 
degenerate cases $\epsilon_i= \epsilon_j$ and $\epsilon_i=\epsilon_k$, henceforth in this section we will assume as hypothesis the 1-structural stability of $\rho$.

\begin{corol}
	\label{prep2_stab} Let $\epsilon_k < \epsilon_j$ be two distinct  eigenvalues of $H$ and 
	$\rho \in \bar{\mathfrak P}_H^{(N,1)}$   a $N$-passive, 1-structurally stable state.
	Then under the condition~(\ref{CONDPREP2}) it holds the implication
	\begin{equation} \label{implic1} 
	\epsilon_k \leq  \epsilon_i \leq \epsilon_j \implies \ln \lambda_i \leq  -\beta (\epsilon_i -\epsilon_k)- \ln Z +  \tfrac{\beta ( \epsilon_j -\epsilon_k)}{N} \; ,
	\end{equation}
	and under either one of the conditions~(\ref{CONDPREP21}) or~(\ref{CONDPREP21_magmin}) it is valid the implication
	\begin{equation}   \label{implic2} 
	\epsilon_k < \epsilon_j \leq  \epsilon_i \Longrightarrow 
	\ln \lambda_i \leq  -\beta(\epsilon_i -\epsilon_k)- \ln {Z} +  \beta\tfrac{(\epsilon_i - \epsilon_k)^2}{N(\epsilon_j - \epsilon_k)} \;. 
	\end{equation} 
\end{corol}
\begin{proof}
	In the cases in which $\epsilon_i$ is neither equal to $\epsilon_j$ or to $\epsilon_k$, (\ref{implic1}) reduces to~(\ref{implic1_inst}), which is implied by condition~(\ref{CONDPREP2}). Likewise, when $\epsilon_i \neq \epsilon_j$, (\ref{implic2}) is equivalent to~(\ref{implic2_inst}), which is implied by~(\ref{CONDPREP21}) or by~(\ref{CONDPREP21_magmin}).
	Consider now the case where $\epsilon_i = \epsilon_j$: under this circumstance
	Eqs.~(\ref{implic1}) 
	trivially follow from (\ref{CONDPREP2}) and by the fact that 
	since $\rho$ is $1$-structurally stable, 
	we must have $\lambda_i= \lambda_j$, see e.g. Eq.~(\ref{EQUAL}).  
	The same reasoning can be applied also for (\ref{implic2}) by exploiting either (\ref{CONDPREP21}) or
	(\ref{CONDPREP21_magmin}). 
	Finally if $\epsilon_i=\epsilon_k$, we must have $\lambda_i=\lambda_k$, and Eq.~(\ref{implic1}) follows again from (\ref{CONDPREP2}). 
\end{proof}

\begin{corol}
	\label{cor3}
	Let $\epsilon_k < \epsilon_j$ be two distinct  eigenvalues of $H$ and
	$\rho \in \bar{\mathfrak P}_H^{(N,1)}$ be a $N$-passive, $1$-structurally stable state.
	Then under the condition~(\ref{CONDPREP21}) of {\bf Proposition \ref{prep2}} 
	we must have 
	\begin{equation}  
	\epsilon_k \leq  \epsilon_i \Longrightarrow \ln \lambda_i \leq -\beta (\epsilon_i-\epsilon_k)  \left( 1 - \tfrac{R(H)}{N} \right) - \ln {Z}\;, \label{IMPO1new} 
	\end{equation} 
	where $R(H)$ is the quantity defined in Eq.~(\ref{DEFR}). 
	Furthermore, if  $\epsilon_k$ is the ground state energy of the system (i.e. $\epsilon_k=\epsilon_0=0$), then the following inequalities hold
	\begin{eqnarray} \label{ENTROINEQ} 
	S(\rho)  &\geq& 
	{\beta} E(\rho;H)  \left( 1 - \tfrac{R(H)}{N} \right) + \ln Z \;, \\ 
	\label{resultprimo} 
	E(\rho;H)   &\leq& \frac{Z_\beta}{Z} \; E_{\beta}(H) e^{\beta \epsilon_{\max} R(H)/N}  \;.
	\end{eqnarray} 
	
\end{corol}
\begin{proof} 
	For $\epsilon_i= \epsilon_k$, Eq.~(\ref{IMPO1new}) is a trivial consequence of the first inequality of assumed
	hypothesis~(\ref{CONDPREP21}).
	For $\epsilon_k <  \epsilon_i \leq \epsilon_j$ instead we notice that
	~(\ref{CONDPREP21}) implies (\ref{CONDPREP2}) and hence 
	~(\ref{implic1}), i.e.
	\begin{equation} \label{implic1here} 
	\ln \lambda_i \leq  -\beta (\epsilon_i -\epsilon_k)- \ln Z +  \tfrac{\beta ( \epsilon_j -\epsilon_k)}{N}
	.\end{equation}
	The thesis then follows from the observation that 
	\begin{eqnarray} 
	\epsilon_j- \epsilon_k = (\epsilon_i-\epsilon_k)  \frac{\epsilon_j- \epsilon_k}{\epsilon_i- \epsilon_k} \leq (\epsilon_i-\epsilon_k)  R(H) \;. 
	\end{eqnarray} 
	For $\epsilon_i >  \epsilon_j$ instead we use the fact that~(\ref{CONDPREP21}) implies 
	~(\ref{implic2}), i.e. 
	\begin{equation}   \label{implic2new} 
	\ln \lambda_i \leq  -\beta(\epsilon_i -\epsilon_k)- \ln {Z} +  \beta\tfrac{(\epsilon_i - \epsilon_k)^2}{N(\epsilon_j - \epsilon_k)} \;, 
	\end{equation} 
	and the inequality  
	\begin{eqnarray} 
	\frac{(\epsilon_i - \epsilon_k)^2}{\epsilon_j - \epsilon_k} = (\epsilon_i-\epsilon_k) \frac{\epsilon_i- \epsilon_k}{\epsilon_j- \epsilon_k} \leq 
	(\epsilon_i-\epsilon_k)  R(H)   \;.
	\label{tuttominR}
	\end{eqnarray} 
	
	Suppose next that $\epsilon_k$ is the ground energy level of the system: in this case the bound
	(\ref{IMPO1new}) on $\lambda_i$ applies for all energy levels, i.e. 
	\begin{equation} 
	- \ln \lambda_i \geq {\beta} \epsilon_i  \left( 1 - \tfrac{R(H)}{N} \right) + \ln Z\;,  \qquad \forall \epsilon_i\;,  \label{IMPO1new2} 
	\end{equation} 
	therefore we can write 
	\begin{eqnarray} 
	S(\rho)  &=& -\sum_i \lambda_i \ln \lambda_i  
	\geq \sum_i \lambda_i {\beta} \epsilon_i \left( 1 - \tfrac{R(H)}{N} \right) +  \ln Z \nonumber  \\
	&=& {\beta} E(\rho;H)  \left( 1 - \tfrac{R(H)}{N} \right) + \ln Z \;.
	\end{eqnarray}
	On the contrary replacing $\epsilon_i \tfrac{R(H)}{N}$ with $\epsilon_{\max} \tfrac{R(H)}{N}$
	in (\ref{IMPO1new2}) we get 
	\begin{equation} 
	\lambda_i \leq Z_{\beta}^{-1} e^{-{\beta} \epsilon_i } e^{\beta \epsilon_{\max} \tfrac{R(H)}{N}}   \frac{Z_\beta}{Z} \;,  \qquad \forall \epsilon_i\;,  \label{IMPO1new222} 
	\end{equation} 
	and hence 
	\begin{eqnarray}
	E(\rho;H) =  \sum_i \lambda_i \epsilon_i\leq  E_{\beta}(H)  e^{\beta \epsilon_{\max} \tfrac{R(H)}{N}} \frac{Z_\beta}{Z} \;. 
	\end{eqnarray} 
\end{proof} 

\begin{corol}
	\label{cor3bis}
	Let $\epsilon_k < \epsilon_j$ be two distinct  eigenvalues of $H$ and
	$\rho \in {\mathfrak P}_H^{(N)}$ be a $N$-passive state.
	Then under the condition~(\ref{CONDPREP21_magmin}) of {\bf Proposition \ref{prep2}} 
	we must have 
	\begin{equation}  
	\epsilon_j <  \epsilon_i \Longrightarrow \ln \lambda_i \leq -\beta (\epsilon_i-\epsilon_k)  \left( 1 - \tfrac{R(H)}{N} \right) - \ln {Z}\;, \label{IMPO1new_casomagZ} 
	\end{equation} 
	where $R(H)$ is the quantity defined in Eq.~(\ref{DEFR}). 
\end{corol}
\begin{proof} 
	The assumed hypothesis~(\ref{CONDPREP21_magmin}) implies~(\ref{implic2_inst}). The thesis follows from~(\ref{implic2_inst}), with the observation~(\ref{tuttominR}).
\end{proof}

The inequalities (\ref{implic1}) and (\ref{implic2}) are the key ingredient to 
arrive at our main result. 
To complete our derivation of equation~(\ref{RESULT}), we find it useful to 
first state a set of hypotheses under which, exploiting {\bf Proposition \ref{prep2}}, one can 
guarantee that Eq.~(\ref{RESULT}) is true
(see {\bf Proposition \ref{ipotesiminimali}} below). Then we show that these conditions are always satisfied by all the elements of $\bar{\mathfrak P}_H^{(N,1)}$ (see 
{\bf Proposition \ref{prepFIN}}). 
\begin{prep}
	\label{ipotesiminimali}
	Let $\rho\in \bar{\mathfrak P}_H^{(N,1)}$ be a $N$-passive, $1$-structurally stable state of a  Hamiltonian $H$
	characterised  by at least three distinct eigenvalues. Then its mean energy energy $E(\rho; H)$  is bounded by the inequalitity (\ref{RESULT}) if
	at least one of the two following conditions holds true for
	the spectrum of $\rho$:
	
	{\bf Condition 1:} there exist energy levels  $\epsilon_{c}> \epsilon_b > 0$ such that 
	\begin{eqnarray}
	\lambda_0 \leq Z_{\beta(\rho)}^{-1}\;, \; \lambda_b    
	\geq  Z_{\beta(\rho)}^{-1} e^{ - \beta(\rho) \epsilon_b} \;, \;
	\lambda_{c}  \leq  Z_{\beta(\rho)}^{-1} e^{ - \beta(\rho) \epsilon_{c}}\;, \nonumber \\
	\label{COND1nuova} 
	\end{eqnarray} 
	with $Z_{\beta(\rho)}$ is the partition function of the isoentropic Gibbs state $\omega_{\beta(\rho)}$, or
	
	{\bf Condition 2:} the population of the ground state is lower bounded by $Z_{\beta(\rho)}^{-1}$, i.e. 
	\begin{eqnarray} 
	\lambda_0 \geq Z_{\beta(\rho)}^{-1}\;. 
	\end{eqnarray} 
\end{prep}
\begin{proof}
	When  {\bf Condition 1} applies we
	notice that the hypotheses (\ref{CONDPREP2}) of {\bf Proposition \ref{prep2}} are fulfilled by 
	identifying $\epsilon_k$ and $\epsilon_j$   
	with $\epsilon_0=0$ and $\epsilon_{c}$ respectively, and taking
	$Z=Z_{\beta(\rho)}$, $\beta=\beta(\rho)$. 
	Accordingly we can then invoke (\ref{implic1}) to establish that for $\epsilon_i$ such that 
	$0\leq  \epsilon_i \leq  \epsilon_{c}$ the following inequality must hold
	\begin{eqnarray} \label{implic1nuova} 
	\ln \lambda_i &\leq&  -\beta(\rho) \epsilon_i - \ln Z_{\beta(\rho)} +  \tfrac{\beta(\rho) \epsilon_{c} }{N} \;,
	\end{eqnarray} 
	which in particular implies 
	\begin{eqnarray} \label{import0} 
	\ln \lambda_i &\leq&  -\beta(\rho) \epsilon_i(1-\tfrac{R(H)}{N})  - \ln Z_{\beta(\rho)} \quad \left( \forall \epsilon_i \leq \epsilon_c \right) \;, \\
	\end{eqnarray}
	(for $0<\epsilon_i\leq \epsilon_{c}$ this follows by the fact  that $R(H) \geq \epsilon_{c}/\epsilon_i$, 
	for 
	$\epsilon_i=0$ instead (\ref{import0}) is a trivial consequence of the first inequality of (\ref{COND1nuova})).
	Now we can apply (\ref{implic1nuova}) for $\epsilon_i=\epsilon_b$ to establish that
	\begin{eqnarray} \label{implic1nuova1b} 
	\lambda_b &\leq& Z'^{-1}:= Z_{\beta(\rho)}^{-1} 
	e^{ -\beta(\rho) \epsilon_b +\beta(\rho) \epsilon_{c}/N } \;.
	\end{eqnarray} 
	Notice also that from (\ref{COND1nuova}) we have 
	\begin{eqnarray} \label{implic1nuova1} 
	\lambda_{c} &\leq& Z_{\beta(\rho)}^{-1} e^{ - \beta(\rho) \epsilon_{c}} \leq \lambda_b
	e^{ -\beta(\rho) (\epsilon_{c}-\epsilon_b)}\;.
	\end{eqnarray} 
	Therefore we notice that the conditions~(\ref{CONDPREP21}) of {\bf Proposition \ref{prep2}} 
	are fulfilled by taking  $\epsilon_k$ and $\epsilon_j$  
	with $\epsilon_b$ and $\epsilon_{c}$ respectively, and taking
	$Z=Z'$ and $\beta=\beta(\rho)$.
	Hence  invoking (\ref{implic2}) we can claim that 
	for all $\epsilon_i\geq \epsilon_{c}$, we must have 
	\begin{eqnarray}  
	\ln \lambda_i &\leq&  -\beta(\rho)(\epsilon_i -\epsilon_b)- \ln {Z'} +  \beta(\rho)\tfrac{(\epsilon_i - \epsilon_b)^2}{N(\epsilon_{c} - \epsilon_b)} \nonumber  \\
	&=&  -\beta(\rho) \epsilon_i - \ln {Z_{\beta(\rho)}} + \beta(\rho) \tfrac{\epsilon_{c}}{N} + \beta(\rho)\tfrac{(\epsilon_i - \epsilon_b)^2}{N(\epsilon_{c} - \epsilon_b)}  \nonumber  \\
	&\leq&  -\beta(\rho) \epsilon_i - \ln {Z_{\beta(\rho)}} + \beta(\rho) \epsilon_b \tfrac{R(H)}{N} + \beta(\rho)\tfrac{(\epsilon_i - \epsilon_b)R(H)}{N} \nonumber   \\
	&\leq&  -\beta(\rho) \epsilon_i (1-\tfrac{R(H)}{N}) - \ln {Z_{\beta(\rho)}} \quad \left( \forall \epsilon_i \geq \epsilon_c \right)  \;, \label{import1} 
	\end{eqnarray} 
	where in the second line we used the definition of $Z'$ given in (\ref{implic1nuova1b}),
	while in the third we used the fact that $R(H)\geq \epsilon_{c}/\epsilon_b$ and 
	$R(H)\geq (\epsilon_i - \epsilon_b)/(\epsilon_{c} - \epsilon_b)$.
	To summarize, under {\bf Condition 1} equations ~(\ref{import0}) (which is valid for every $\epsilon_i \leq \epsilon_c$) and (\ref{import1}) (which is valid for every $\epsilon_i \geq \epsilon_c$) establish that
	\begin{eqnarray} \label{implic1nuovasum} 
	- \ln \lambda_i &\geq& \beta(\rho) \epsilon_i (1-\tfrac{R(H)}{N}) + \ln {Z_{\beta(\rho)}} \;, \quad \forall \epsilon_i\;.
	\end{eqnarray} 
	Hence we can write 
	\begin{eqnarray} 
	S(\rho)  &= &  -\sum_i \lambda_i \ln \lambda_i \\ 
	& \geq & \sum_i \lambda_i {\beta(\rho)} \epsilon_i \left( 1 - \tfrac{R(H)}{N} \right) +  \ln Z_{\beta(\rho)} \nonumber   \\
	&=& {\beta}(\rho) E(\rho;H)  \left( 1 - \tfrac{R(H)}{N} \right) + \ln Z_{\beta(\rho)} \;.
	\end{eqnarray} 
	which finally leads to (\ref{RESULT}) by using~(\ref{sameE}) and (\ref{DEFENTROBETA}) to enforce the identity
	\begin{eqnarray} 
	E_{\beta(\rho)}(H) = \frac{S(\rho)-\ln Z_{\beta(\rho)
	}}{\beta(\rho)}\;. \label{ide11} 
	\end{eqnarray} 
	Consider next the case where  {\bf Condition 2} holds. Under this circumstance let us introduce
	the Gibbs state $\omega_{\beta'}$ with inverse temperature $\beta' \geq \beta(\rho)$, such that 
	\begin{eqnarray}\lambda_0=Z_{\beta'}^{-1} \label{defbeta'}\;. \end{eqnarray} Because $\Tr \rho = \Tr \omega_{\beta'} = 1$, there exists at least one eigenvalue $\epsilon_{c}>0$  such that 
	\begin{eqnarray}
	\lambda_{c} \leq \lambda_{c}' = Z_{\beta'}^{-1}e^{-\beta' \epsilon_{c}}=\lambda_0e^{-\beta' \epsilon_{c}}\;,  \label{deflamdaa} 
	\end{eqnarray}
	$\lambda_{c}'$ being the associated
	population of the Gibbs state $\omega_{\beta'}$. 
	Accordingly  the hypothesis of {\bf Corollary \ref{cor3}} are
	fulfilled with $\epsilon_k=\epsilon_0=0$, $\epsilon_j=\epsilon_{c}$, $Z=Z_{\beta'}$, and $\beta=\beta'$. 
	Hence invoking (\ref{ENTROINEQ}) we can write 
	\begin{equation} 
	\beta(\rho) E_{\beta(\rho)}(H)+\ln Z_{\beta(\rho)}   \geq 
	{\beta}'  E(\rho;H)  \left( 1 - \tfrac{R(H)}{N} \right) + \ln Z_{\beta'}\label{lesserbeta}
	\end{equation} 
	where in the left-hand-side we used (\ref{ide11}) to express $S(\rho)$ in terms of $E_{\beta(\rho)}(H)$.
	In case $\beta'= \beta(\rho)$ Eq.~(\ref{lesserbeta}) is just (\ref{RESULT}) and the proof ends. On the contrary 
	if $\beta' > \beta(\rho)$, exploiting  the fact that function $\ln Z_\beta$ is convex we can write
	\begin{equation}
	\ln Z_{\beta'} \geq \ln Z_{\beta(\rho)} + \Delta\beta \frac{\partial}{\partial \beta} \ln Z_\beta\Bigr|_{\beta ={\beta(\rho)}}
	 = \ln Z_{{\beta(\rho)}} - E_{{\beta(\rho)}}\Delta\beta\;,
	\label{lnZconcava}
	\end{equation}
	with $\Delta \beta:= \beta' - {\beta(\rho)}> 0$. 
	Using the inequality ~(\ref{lnZconcava}) into ~(\ref{lesserbeta}), and rearranging the terms, we hence arrive to
	\begin{equation}
	\beta'  E_{\beta(\rho)}(H) \geq \beta' \left( 1 - \tfrac{R(H)}{N} \right)  E(\rho;H) \label{consequentiamirabilis}\;,
	\end{equation}
	which leads to (\ref{RESULT}) by the strict positivity of $\beta'$. 
\end{proof}
\begin{prep} \label{prepFIN}
	The mean energy energy $E(\rho; H)$ of a $N$-passive, $1$-structurally stable state 
	$\rho\in \bar{\mathfrak P}_H^{(N,1)}$ of a Hamiltonian $H$
	characterised  by at least three distinct eigenvalues, 
	is bounded by the inequality~(\ref{RESULT}). 
\end{prep}
\begin{proof}  
	If $\rho$ is such that $\lambda_0 \geq Z_{\beta(\rho)}^{-1}$ then  the thesis follows by application of 
	{\bf Condition 2} of {\bf Proposition \ref{ipotesiminimali}}.
	Consider next the case $\lambda_0 < Z_{\beta(\rho)}^{-1}$. Here we observe that 
	since the von Neumann entropies of $\rho$ and $\omega_{\beta(\rho)}$ coincide, i.e. $S(\rho) = S(\omega_{\beta(\rho)})$, the spectrum of  $\rho$ does not strictly majorize, nor is strictly majorized by the 
	spectrum  of $\omega_{\beta(\rho)}$~\cite{MAJ0,MAJ}. As shown in   Appendix~\ref{maj}
	we can then claim that there must exist $\epsilon_{c}> \epsilon_b > 0$ such that 
	$\lambda_b \geq  \hat{\lambda}_b$ and
	$\lambda_{c} \leq  \hat{\lambda}_{c}$, where
	\begin{eqnarray} \label{DEFLAMBDAHAT} 
	\hat{\lambda}_j:= Z_{\beta(\rho)}^{-1} e^{ - \beta(\rho) \epsilon_j}\;,\end{eqnarray} are the eigenvalues of 
	the Gibbs state $\omega_{\beta(\rho)}$.
	Therefore this time the thesis derives
	as a consequence of {\bf Condition 1} of {\bf Proposition \ref{ipotesiminimali}}.
\end{proof}

\begin{prep} \label{prepFINexp}
	The mean energy energy $E(\rho; H)$ of a $N$-passive, $1$-structurally stable state 
	$\rho\in \bar{\mathfrak P}_H^{(N,1)}$ of a Hamiltonian $H$
	characterised  by at least three distinct eigenvalues, 
	is bounded by the inequality~(\ref{RESULTexp}).
	
	Furthermore, the inequality~(\ref{RESULTexp}) also holds if $\rho\in {\mathfrak P}_H^{(N)}$ is a generic $N$-passive state, provided that the populations of $\rho$ satisfy the condition 
	\begin{equation}
	\label{COND00}
	\epsilon_i = 0 \implies \lambda_i \geq Z_{\beta(\rho)}^{-1}\;. 
	\end{equation}
\end{prep}
\begin{proof}
	Consider first the case in which the population of the ground state of $\rho\in \bar{\mathfrak P}_H^{(N,1)}$ is $\lambda_0 \leq Z_{\beta(\rho)}^{-1}$. In this case, the majorization argument of Appendix~\ref{maj} implies the validity of~(\ref{COND1nuova}), that is, that there exist  $\epsilon_{c}> \epsilon_b > 0$ such that 
	$\lambda_b \geq  \hat{\lambda}_b$ and
	$\lambda_{c} \leq  \hat{\lambda}_{c}$ (again we use the convention~(\ref{DEFLAMBDAHAT}) 
	to indicate the eigenvalues of 
	the Gibbs state $\omega_{\beta(\rho)}$). 
	As seen during the proof of \textbf{Proposition \ref{ipotesiminimali}},~(\ref{COND1nuova}) implies the inequality~(\ref{implic1nuovasum}).
	We now observe that replacing $\epsilon_i\tfrac{R(H)}{N}$ in the right-hand-side of
	(\ref{implic1nuovasum}) with $\epsilon_{\max} \tfrac{R(H)}{N}$, we get 
	\begin{eqnarray} \label{implic1nuovasumnuova} 
	\lambda_i &\leq&Z_{\beta(\rho)}^{-1} e^{- \beta(\rho) \epsilon_i} e^{ \beta(\rho) \epsilon_{\max} \tfrac{R(H)}{N}} \;, \quad \forall \epsilon_i\;,
	\end{eqnarray} and hence 
	\begin{eqnarray}
	E(\rho;H) =  \sum_i \lambda_i \epsilon_i\leq  E_{\beta(\rho)}(H)  e^{\beta(\rho) \epsilon_{\max} \tfrac{R(H)}{N}}  \;. 
	\end{eqnarray} 
	
	Next we consider the case in which $\rho\in {\mathfrak P}_H^{(N)}$ is a $N$-passive state, not necessarily included in $\bar{\mathfrak P}_H^{(N,1)}$, and~(\ref{COND00}) holds. When also $\rho \in \bar{\mathfrak P}_H^{(N,1)}$, by virtue of equation~(\ref{EQUAL}) the condition~(\ref{COND00}) is implied by $\lambda_0 \geq Z_{\beta(\rho)}^{-1}$; therefore, this case will also complete the proof of~(\ref{RESULTexp}) for all the $N$-passive, 1-structurally stable states.
	Since $\rho$ and $\omega_{\beta(\rho)}$ have both trace one, there must exist some other eigenvalue $\lambda_{b}$ of $\rho$ such that $\lambda_{b} \leq \hat{\lambda}_{b} = Z_{\beta(\rho)}^{-1}e^{-\beta \epsilon_{b}}$. The assumption~(\ref{COND00}) ensures that $b>d_0$, i.e., that $\epsilon_b > 0$. Choose $\lambda_{b}$ to be the first eigenvalue of $\rho$ which is smaller than the corresponding eigenvalue of $\omega_{\beta(\rho)}$, i.e. $\lambda_b \leq \hat{\lambda}_b$ and $\lambda_i \geq \hat{\lambda}_i$ for every $i < b$ (it is worth remarking that, in presence of an Hamiltonian with a degenerate spectrum, the set of indices such that $\epsilon_i < \epsilon_{b}$ does not coincide with the set of indices such that $i < b$. In what follows we will make use of both kinds of conditions). This implies 
	\begin{equation} \label{iniziosopra}
	\sum_{i < b} \lambda_i \geq \sum_{i < b} \hat{\lambda}_i.
	\end{equation}
	The hypotheses (\ref{CONDPREP21_magmin}) of {\bf Proposition \ref{prep2}} are fulfilled by  identifying $\epsilon_k$ and $\epsilon_j$   
	respectively with $\epsilon_0=0$ and $\epsilon_{b}$, and taking
	$Z=Z_{\beta(\rho)}$, $\beta=\beta(\rho)$. 
	Therefore, for any $\epsilon_i > \epsilon_{b}$ the inequality~(\ref{IMPO1new_casomagZ}) holds:
	\begin{equation}  
	\epsilon_i > \epsilon_{b} \Longrightarrow \ln \lambda_i \leq -\beta(\rho) \epsilon_i  \left( 1 - \tfrac{R(H)}{N} \right) - \ln {Z_{\beta{\rho}}}\;. 
	\label{controlloquellidopo}
	\end{equation} 
	Given the non-majorization condition, we can infer from~(\ref{iniziosopra}) that there must exist some other eigenvalue $\lambda_c$, with $c>b$, such that 
	\begin{equation} \label{sorpasso_sum}
	\sum_{i < c} \lambda_i \leq \sum_{i < c} \hat{\lambda}_i
	\end{equation}
	Now we claim that we can choose an index $c$ such that $\epsilon_{c+1} > \epsilon_{b}$. Indeed, recalling that the eigenvalues $\{ \lambda_i \}_i$ are arranged in decreasing order, if $\epsilon_{c+1} = \epsilon_b$ then $\lambda_{c+1} \leq \lambda_{b} \leq \hat\lambda_b = \hat\lambda_{c+1}$, and equation~(\ref{sorpasso_sum}) continues to be valid if we extend the range of the sum~(\ref{sorpasso_sum}) to $i < c+1$. Therefore, we can always choose an index $c$ in~(\ref{sorpasso_sum}) such that $\epsilon_{c+1} > \epsilon_c$. So the inequality~(\ref{controlloquellidopo}) is valid for every $i > c$, since every $i > c$ also satisfies $\epsilon_i \geq \epsilon_{c+1} > \epsilon_c$.
	From the equations~(\ref{iniziosopra}) and~(\ref{sorpasso_sum}) follows that
	\begin{equation}
	\sum_{i < b} \left( \lambda_i - \hat{\lambda}_i \right) \leq - \sum_{b \leq i < c} \left( \lambda_i - \hat{\lambda}_i \right)\;,
	\label{sorpasso}
	\end{equation}
	which in turns implies that
	\begin{equation}
	\sum_{i < b} \epsilon_i \left( \lambda_i - \hat{\lambda}_i \right) < \sum_{i < b} \epsilon_{b} \left( \lambda_i - \hat{\lambda}_i \right)  
	 \leq - \sum_{b \leq i < c} \epsilon_{b} \left( \lambda_i - \hat{\lambda}_i \right) <
	- \sum_{b \leq i < c} \epsilon_i \left( \lambda_i - \hat{\lambda}_i \right) 
	\end{equation}
	or equivalently that
	\begin{equation}\label{raddrizza_inizio}
	\sum_{i \leq c} \left( \lambda_i - \hat{\lambda}_i \right) \epsilon_i < 0 \; .
	\end{equation}
	Using~(\ref{raddrizza_inizio}) between the second and the third line and then~(\ref{controlloquellidopo}) in the third line, we finally conclude that 
	\begin{multline}
	E(\rho;H) = \sum_i \lambda_i \epsilon_i = \sum_{i \leq c} \lambda_i \epsilon_i + \sum_{i > c} \lambda_i \epsilon_i = \\
	\sum_{i \leq c} \hat{\lambda}_i \epsilon_i + \sum_{i \leq c} \left( \lambda_i - \hat{\lambda}_i \right) \epsilon_i  + \sum_{i > c} \lambda_i \epsilon_i \\
	<  \sum_{i \leq c} \hat{\lambda}_i \epsilon_i + \sum_{i > c} \lambda_i \epsilon_i \leq 
	\sum_i \hat{\lambda}_i e^{ \beta(\rho) \epsilon_{i} \tfrac{R(H)}{N}}  \epsilon_i \\
	< e^{ \beta(\rho) \epsilon_{\max} \tfrac{R(H)}{N}} \sum_i \hat{\lambda}_i \epsilon_i = e^{ \beta(\rho) \epsilon_{\max} \tfrac{R(H)}{N}} E_{\beta(\rho)} \;,
	\end{multline}
	which proves the thesis. 	
\end{proof}

\subsection{Saturation of the inequality~(\ref{RESULT})
}\label{sec:sat} 

\begin{figure}
	\centering
	\includegraphics[width=\columnwidth]{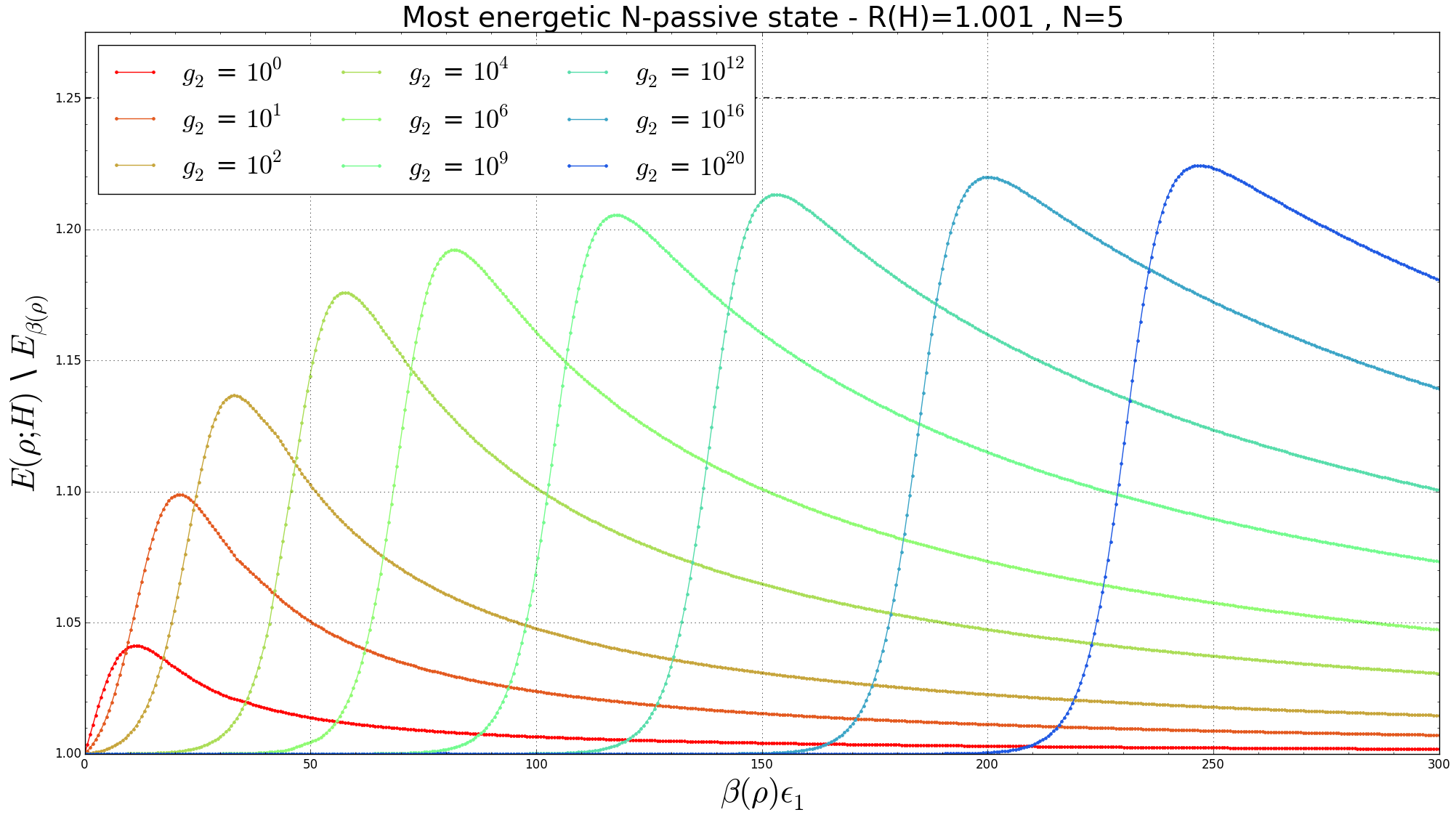}
	\caption{Energy ratio $\alpha := E(\rho; H) / E_{\beta(\rho)}$ of the most energetic 5-passive, \ 1-structurally stable state as a function inverse temperature $\beta(\rho)$, for an Hamiltonian with three energy levels $\epsilon_0 = 0$, $\epsilon_1$, and $\epsilon_2 = 1.001 \epsilon_1$, ($r=1.001$)  and for different values of the degeneracy $g_2$. The  dashed line at the top of the plot represents the upper bound (\ref{alphamax}).	 }
	\label{fig:moltodegeneri}
\end{figure}

The inequality~(\ref{RESULT}) implies that 
for $N>R(H)$, the ratio 
\begin{eqnarray} \label{DEFALPHA}  \alpha:= E(\rho; H) / E_{\beta(\rho)} (H)\;,\end{eqnarray}  between the mean energy of a $N$-passive, $1$-structurally stable  state $\rho$ and the energy of the Gibbs state $\omega_{\beta(\rho)}$ that has the same entropy of $\rho$, can be at most
\begin{equation}
\label{alphamax}
\alpha \leq  \alpha_{\max} := \frac{N}{N-R(H)} \;. 	\end{equation}
In this section we shall exhibit an explicit example of $N$-passive, $1$-structurally stable states whose energy are arbitrarily close to the limit (\ref{alphamax}) -- see also Fig.~\ref{fig:moltodegeneri}.
Although this example is contrived (it requires very small temperatures and very large degeneracies), it works for all $N\geq 2$, for some specific value of $\alpha_{\max}$. In particular our example requires that 
for some integer $1 \leq m < N$ one has
\begin{equation}
\label{CONDsuR}
R(H)/N \gtrsim \frac{1}{m+1}\;, 
\end{equation} 
or equivalently $\alpha_{\max} \gtrsim\frac{m+1}{m}$.
Given $N\geq 2$, consider the
set of 	$N$-passive, $1$-structurally stable state $\rho\in \bar{\mathfrak P}_H^{(N,1)}$ 
associated with a Hamiltonian $H$ characterised by three energy levels: 
\begin{eqnarray}
\epsilon_0 = 0\;, \quad \epsilon_1>0\;, 
\quad \epsilon_2 = r \epsilon_1\;,
\end{eqnarray}
with $r>1$ being the parameter that provides the $R(H)$ of the model,
i.e. 
\begin{eqnarray} R(H) = r \;. \label{RgandeugualeRpiccolo}
\end{eqnarray} 
We  also assume $\epsilon_1$ and $\epsilon_2$ to have degeneracies $g_1$ and $g_2$ respectively,
whose values will be specified later on. 
Under the above premise, in what follows we shall use $\epsilon_1$, $g_1$, and $g_2$ 
as free parameters over which we optimize to enforce the saturation of the bound (\ref{alphamax}).

First of all, given 
the associated
iso-entropic Gibbs counterpart $\omega_{\beta(\rho)}$ of  $\rho\in \bar{\mathfrak P}_H^{(N,1)}$
we exploit $\epsilon_1$ to force it into the
low temperature
regime imposing the constraint 
\begin{eqnarray}
\beta(\rho) \epsilon_1 \gg 1\;. \label{lowt}  
\end{eqnarray} 
On one hand this assumption 
makes sure that  (\ref{RESULT}) 
provides a  bound which is tighter than the exponential one given by Eq.~(\ref{RESULTexp}). 
On the other hand, we can use (\ref{lowt})  to approximate the populations
$\hat{\lambda}_j= Z_{\beta(\rho)}^{-1} e^{ - \beta(\rho) \epsilon_j}$ 
of $\omega_{\beta(\rho)}$ as 
\begin{equation}
\label{lambdaeq}
\hat{\lambda}_0 \simeq 1,  \quad  \hat{\lambda}_1
\simeq e^{-\beta(\rho)\epsilon_1}, \quad \hat{\lambda}_2 \simeq  e^{-\beta(\rho) r \epsilon_1} \;.
\end{equation}
For a reason that will soon become apparent, we impose an additional condition on $\epsilon_1$, namely that
\begin{equation}
\label{lowlowt}
\beta(\rho) \epsilon_1 > \frac{N}{r} \ln (r \alpha_{\max}) \;.
\end{equation}
This condition is clearly compatible with~(\ref{lowt}).

Now we fix $g_1$ and $g_2$ to ensure that despite the fact that 
$\hat{\lambda}_1 \gg \hat{\lambda}_2$, the total population in the $\epsilon_2$ energy level will  still be bigger than the total population in the $\epsilon_1$, i.e.  
we assume 
\begin{equation}
\label{rososkins}
g_2 e^{-\beta(\rho) r \epsilon_1 } \gg g_1 e^{-\beta(\rho) \epsilon_1}  \Longrightarrow
\frac{\ln ({g_2}/{g_1})}{\beta(\rho) \epsilon_1}   > r - 1\;. 
\end{equation}
Condition (\ref{rososkins}) ensures that the energy and the entropy of the thermal equilibrium states $\omega_\beta(\rho)$ are dominated by the contribution from the higher energy level of the system, i.e. 
\begin{eqnarray}
\label{E_0}
E_{\beta(\rho)}(H)  &\simeq& g_2 e^{-\beta(\rho) \epsilon_1 r} \epsilon_1 r\;, \\ 
\label{S_0} 
S_{\beta(\rho)} &\simeq& g_2 e^{-\beta(\rho) \epsilon_1 r} \beta(\rho) \epsilon_1 r\;.
\end{eqnarray}
Take now $m < N$ integer such that  
\begin{equation}
\frac{m}{N} <\frac{1}{r} \leq \frac{m+1}{N} \;. \label{asdf}
\end{equation}
Note  that such $m$ can be identified with the same $m$ of (\ref{CONDsuR}):
accordingly to fully comply with such constraint we require 
$1/r$ to be very close to the upper bound of~(\ref{asdf}), so that 
we can also write 
\begin{equation}
\frac{m}{N} = \frac{m+1}{N} - \frac{1}{N} \gtrsim \frac{1}{r} - \frac{1}{N}\;.
\label{massimadistanza}
\end{equation}
By virtue of \textbf{Proposition \ref{PREP1}}, equation~(\ref{asdf}) implies that for $\rho\in \bar{\mathfrak P}_H^{(N,1)}$ there must hold the inequality 
\begin{eqnarray} \label{constraint} 
\lambda_1 \leq  \lambda_2^{\frac{m}{N}}\lambda_0^{\frac{N-m}{N}}\lesssim   \lambda_2^{\frac{N-r}{rN}}\;,\end{eqnarray}
where in the last passage we used equation~(\ref{massimadistanza}) and the fact that $\lambda_0\lesssim 1$. 
Now we focus on the special subset of the density matrices $\rho\in \bar{\mathfrak P}_H^{(N,1)}$ that have $\lambda_0\simeq 1$  and 
$\lambda_1 \ll 1$, and which saturate the limit posed by 
Eq.~(\ref{constraint}). We parametrize the populations of $\rho$ as 
\begin{equation}
\label{lambdamep}
\lambda_0 \simeq 1, \quad   \lambda_1 = \xi, \quad \lambda_2 \simeq
(\eta \xi)^{\frac{rN}{N-r}} =  (\eta\xi)^{r \alpha_{\max}}\;, 
\end{equation}
with $\eta\leq 1$ and  $\xi\ll 1$ which in particular 
we assume to fulfil the inequality 
\begin{equation}
\label{condxi1}
\xi \ll 
e^{-k_0\frac{(r -1) \beta(\rho) \epsilon_1}{ r \alpha_{\max} -1}} \;,
\end{equation}
with  $k_0\gg 1$ being some large fixed constant (notice  that thanks to 
(\ref{lowt}), Eq.~(\ref{condxi1}) is in perfect agreement with the request of having $\xi\ll 1$, indeed
the larger is $\beta(\rho) \epsilon_1$ the smaller is $\xi$). 
We now impose  the energy and the entropy of the state (\ref{lambdamep}) to be dominated by 
$\lambda_1$. Accordingly
we set a new condition for $g_1$ and $g_2$, requiring that 
\begin{equation}
\label{condxi}
g_1 \xi \gg  g_2 \xi^{r \alpha_{\max}}  \Longrightarrow
\frac{\ln ({g_2}/{g_1})}{\beta(\rho) \epsilon_1}   < \frac{1-r \alpha_{\max}}{ \beta(\rho) \epsilon_1} \ln \xi \;,
\end{equation}
which thanks to our choice (\ref{condxi1}) is perfectly compatible with our previous assumption 
(\ref{rososkins}).

Equations~(\ref{condxi}) and (\ref{lambdamep}) lead to
the following approximated expressions for the mean energy and entropy of~$\rho$:
\begin{equation}
\label{E_P}
E (\rho; H) \simeq g_1 \xi \epsilon_1\;,  \qquad  
S (\rho) \simeq -g_1\xi \ln \xi\;.
\end{equation}
Equation~(\ref{E_P}) gives the energy $E(\rho; H)$ and the entropy $S(\rho)$ of the state $\rho$ (whose populations are defined in~(\ref{lambdamep})), as a function of the parameters $g_1$, $\xi$ and $\epsilon_1$. Equations~(\ref{E_0}) and~(\ref{S_0}) express respectively the energy $E_\beta(H)$ and the entropy $S_\beta$ of Gibbs states, as a function of the parameters $g_2$, $r$ and $\epsilon_1$. We are now ready to solve them with the conditions $S(\rho) = S\beta$ and $E(\rho; H) = \alpha E_\beta$.

On the one hand, together with~(\ref{E_0}),the first expression of ~(\ref{E_P}) allows us to write the ratio~(\ref{DEFALPHA}) in terms of $\xi$ as 
\begin{eqnarray}
\label{ansatzxi}
\alpha &\simeq&\xi  \frac{g_1}{g_2}  \frac{ e^{\beta(\rho) \epsilon_1 r}}{r} \;.
\end{eqnarray}
On the other hand, from the second expression of Eq.~(\ref{E_P}), we obtain
the additional condition 
\begin{equation}
\label{equiS}
g_2 e^{-\beta(\rho) \epsilon_1 r} \beta(\rho) \epsilon_1 r \simeq -g_1 \xi \ln \xi \;,
\end{equation}
that follows from the request that  $\rho$ and $\omega_{\beta{(\rho)}}$ have the same entropy
(once more it is worth noticing that no conflict arises with our previous assumptions, since
the large values of $\beta(\rho) \epsilon_1$ imposed by (\ref{lowt}) are in agreement with small
values of $\xi$). 
Replacing  (\ref{ansatzxi}) for $\xi$ into  (\ref{equiS}) leads to a transcendental equation
for  the ratio~(\ref{DEFALPHA})  of the model: 
\begin{equation} \label{ALPHAtrans} 
\alpha^{-1} \simeq  r - \frac{\ln{( g_2 / g_1 )}}{\beta(\rho) \epsilon_1} - \frac{\ln(r\alpha)}{\beta(\rho) \epsilon_1} \;.
\end{equation}
We now claim that  it is possible to set the parameters of the model (i.e. the quantities
$r$, $\epsilon_1$, $g_1$, $g_2$, and $k_0$) in such a way that
the bound (\ref{RESULT}) saturates, 
by forcing the solution 
$\alpha \simeq \alpha_{\max}$ from  Eq.~(\ref{ALPHAtrans}) while 
fulfilling all  the constraints we invoked in the derivation, i.e. the inequalities (\ref{CONDsuR}), (\ref{lowt}), (\ref{lowlowt}), (\ref{rososkins}), and (\ref{condxi1}).

To see this let first observe that 
from (\ref{rososkins}) it follows that $\alpha^{-1} < 1    - \frac{\ln(r\alpha)}{\beta(\rho) \epsilon_1} $ 
which simply says that $\alpha$ is a quantity greater than 1. On the contrary 
a lower bound for $\alpha^{-1}$ can be obtained 
via the constraint~(\ref{condxi1}) which via (\ref{ansatzxi}) can be written as 
$\frac{\ln(g_2/g_1)}{\beta(\rho)\epsilon_1} \ll r - \frac{r-1}{r  \alpha_{\max} - 1}	-  \frac{\ln(r \alpha)}{\beta(\rho) \epsilon_1}$. Inserting this into  (\ref{ALPHAtrans}) this implies 
$\alpha^{-1} \gg  \frac{r-1}{r  \alpha_{\max} - 1}$,
whose right-hand-side is strictly smaller than $\alpha_{\max}^{-1}$ due to the fact that $\alpha_{\max}\geq 1$ by construction.
Therefore, as far as it concerns to (\ref{rososkins}) and (\ref{condxi1}), 
$\alpha_{\max}$  is   inside of the domain of the allowed
values of $\alpha$ obtainable when solving~(\ref{ALPHAtrans}). 
To check the compatibility of such result with (\ref{CONDsuR}) and (\ref{lowt}) let us solve  (\ref{ALPHAtrans})
for $\frac{\ln(g_2/g_1)}{\beta(\rho)\epsilon_1}$ when 
$\alpha$ is taken to be equal to $\alpha_{\max}$, i.e. 
\begin{eqnarray}
\frac{\ln(g_2/g_1)}{\beta(\rho)\epsilon_1} \simeq r - \alpha_{\max}^{-1}	-  \frac{\ln(r \alpha_{\max})}{\beta(\rho) \epsilon_1}
\end{eqnarray} 
which to be in agreement with (\ref{rososkins})  would require 
\begin{eqnarray}
\alpha_{\max}^{-1}	+ \frac{\ln(r \alpha_{\max})}{\beta(\rho) \epsilon_1} < 1\;,
\end{eqnarray} 
a condition which is equivalent to~(\ref{lowlowt}).

\section{Non structurally stable, $N$-passive states} \label{sec:struct}

The bounds derived in the previous section in general do not apply to states which are just  $N$-passive.
Indeed the conditions of \textbf{Proposition \ref{ipotesiminimali}} may not be fulfilled by a passive state which is not structurally stable: one can always find a couple of eigenstates $\lambda_b>\lambda_a$ of $\rho$ such that $\lambda_a < Z_{\beta(\rho)}^{-1}e^{-\beta(\rho)\epsilon_a}$ and $\lambda_b > Z_{\beta(\rho)}^{-1}e^{-\beta(\rho)\epsilon_b}$, but their energies could be equal ($\epsilon_a = \epsilon_b$), and in this case (\ref{RESULT}) or (\ref{RESULTexp})  needs not to be valid. 
Of course this problem may arise only if the spectrum of $H$ is degenerate since, due to 	Eq.~(\ref{EXPINCLUHnondeg}), for non-degenerate Hamiltonians all $N$-passive states are also $N$-passive and $1$-structurally stable and the bounds we have derived trivially hold true. 
At least for the bound~(\ref{RESULTexp})
a similar conclusion can be drawn in the presence of degeneracies of the spectrum of $H$,
for all  $N$-passive state $\rho \in {\mathfrak P}_H^{(N)}$ 
whose ground state populations are larger than or equal to the ground state population of their
associated isoentropic Gibbs states, i.e. when Eq.~(\ref{COND00}) is true: under such condition, by proposition {\bf Proposition \ref{prepFINexp}} a generic $\rho \in {\mathfrak P}_H^{(N)}$ will still respect the bound~(\ref{RESULTexp}) -- see Table~\ref{TABLE1}.

In summary,  the only cases which are left uncovered
by at least one of our two bounds, is when
$H$ is degenerate, 
$\rho \in {\mathfrak P}_H^{(N)}$ is not $1$-structurally  stable and violate
the condition~(\ref{COND00}).
Aim of the present section is to deal with these special configurations.	 To begin with, it is worth remarking that  this case
includes both the situation where $\rho$ has sufficiently large entropy which allows us to identify 
an isoentropic Gibbs counterpart  $\omega_{\beta(\rho)}$, as well as the more pathological cases
where $S(\rho) <\ln d_0$ for which  $\omega_{\beta(\rho)}$ does not even exists. 
Still, in both scenarios   we can associate to  $\rho$ a $N$-passive, $1$-structurally stable density matrix  
$\bar{\rho} \in \bar{\mathfrak P}_H^{(N,1)}$ obtained by 
replacing the populations  $\lambda_j$s  of $\rho$ with
their  mean values computed by averaging them over all the energy levels with the same energy eigenvalues, i.e. 
\begin{equation} 
\bar{\lambda}_i := \frac{1}{d_{\epsilon_i}} \sum_{\epsilon_j = \epsilon_i} \lambda_j \;,
\label{spianatura}
\end{equation}
where $d_{\epsilon_i}$ is the degeneracy of the energy level $\epsilon_i$. 
One can easily verify that the spectrum of $\bar{\rho}$ is majorized by the one of $\rho$~\cite{MAJ0,MAJ}. Therefore, 
while by construction 
$\bar\rho$ has
the same energy as $\rho$, its entropy is certainly not smaller than  $S(\rho)$, i.e.
\begin{eqnarray}  \label{ENTO} 
E(\bar{\rho};H) = E(\rho;H) \;, \qquad S(\bar{\rho}) \geq S(\rho)\;.  \end{eqnarray} 
Furthermore, since  $\bar{\rho}$ is $N$-passive and $1$-structurally stable it respects the inequalities (\ref{RESULTexp}) and (\ref{RESULT}). This means that, given a Hamiltonian with at least three distinct eigenvalues, for any $N > R(H)$ we can write
\begin{equation}
E(\rho;H) = E(\bar{\rho};H) 
\leq E_{{\beta}(\bar{\rho})}(H) \min\left\{ \left( 1 - \tfrac{R(H)}{N} \right)^{-1},e^{\beta(\bar{\rho}) \epsilon_{\max} \frac{R(H)}{N}} \right\} \;,
\label{goulish_approach1}    
\end{equation}
where as usual ${\beta}(\bar{\rho})$ indicates the inverse
temperature of the Gibbs state $\omega_{\beta(\bar{\rho})}$ that has the same entropy of $\bar{\rho}$.
By expanding Eq.~(\ref{goulish_approach1})  at large $N$ we can finally arrive to the following
compact expression  
\begin{equation}
E(\rho;H) = E(\bar{\rho};H) 
\leq E_{{\beta}(\bar{\rho})}(H) \left[ 1 + \frac{u(\bar{\rho})R(H)}{N} +  \mathcal{O} \left( \frac{1}{N^2} \right) \right]  \; ,
\label{goulish_approach}    
\end{equation}
where now 
$u(\bar{\rho}) := \min\{1, \beta(\bar{\rho})\epsilon_{\max}\}$.
Assume next that $S(\rho)\geq \ln d_0$ so that $\omega_{\beta(\rho)}$ does exist.
Notice that by the monotonocity relation that connects  the Gibbs functionals~(\ref{DEFNENBETA0})
and (\ref{DEFENTROBETA}), from (\ref{ENTO}) we have 
$\beta(\bar{\rho}) \leq \beta(\rho)$ and also that  
$E_{{\beta}(\bar{\rho})}(H)$ cannot be smaller than $E_{{\beta}({\rho})}(H)$, i.e. 
\begin{equation} 
S_{\beta(\bar{\rho})} = S(\bar{\rho}) \geq S(\rho) =S_{\beta(
	{\rho})}
 \Longrightarrow \label{ttt} 
E_{{\beta}(\bar{\rho})}(H)\geq E_{{\beta}({\rho})}(H). 
\end{equation}  
In order to convert~(\ref{goulish_approach1}) or (\ref{goulish_approach}) into a bound
that links the energy of $\rho$ with the energy of its Gibbs counterpart we need to find a way 
to reverse the inequality (\ref{ttt}), constructing an upper bound for  $E_{{\beta}({\rho})}(H)$ in terms
of $E_{{\beta}(\bar{\rho})}(H)$. 
For this purpose in the next paragraphs we  determine  
an upper bound of the quantity 
\begin{eqnarray}\Delta S(\rho) := S(\bar{\rho}) - S(\rho)= 
S_{\beta(\bar{\rho})}-S_{\beta(
	{\rho})}\;, \label{DEFINIZIONEDELTAS} 
\end{eqnarray} 
which
using again the monotonocity connection between~(\ref{DEFNENBETA0})
and (\ref{DEFENTROBETA})
will then be converted into the  inequality  we are looking for. 
Our final result will be that, whenever 
the condition~(\ref{COND00}) is false and $S(\rho)\geq \ln d_0$, we can write
\begin{multline}
E(\rho;H)<  \frac{N}{N-2}\left[ 1 + \frac{u(\rho)R(H)}{N} \right] E_{\beta(\rho)}(H) \\ +  \frac{(d_0 - 1)Z_{\beta(\rho)}^{-1}\epsilon_{\max}}{N-2}  + \frac{\beta^{-1}(\rho)}{N-2} + \mathcal{O} \left( \frac{1}{N^2} \right)\;,
\label{RESULT_inst}    
\end{multline}
with 
\begin{eqnarray} u({\rho}) := \min\{1, \beta({\rho})\epsilon_{\max}\}\;. \label{DEFURHO} \end{eqnarray} 
In the case of a Hamiltonian with a non-degenerate ground state ($d_0 = 1$) the above expression can be further
simplified to 
\begin{multline}
E(\rho;H) <
\frac{N}{N-2} \left[ 1 + \tfrac{u(\rho)R(H)}{N} \right]E_{\beta(\rho)}(H) + \tfrac{\beta^{-1}(\rho)}{N-2} +  \mathcal{O} \left( \frac{1}{N^2} \right) 
\label{RESULT_inst1}  \\
= E_{\beta(\rho)}(H) \left[ 1 + \tfrac{u(\rho)R(H)+2}{N} \right] + \tfrac{1}{N\beta(\rho)} + \mathcal{O} \left( \frac{1}{N^2} \right) \;.
\end{multline}
The above expressions refers to all the cases where $H$ has at least three independent eigenvalues. 
The only non-trivial configuration which is left unsolved is the one where $H$ is a two-level Hamiltonian 
and the system has dimension $d\geq 3$.
In this case, we show that (\ref{RESULT_inst}) is replaced by 
\begin{equation}
E(\rho;H)<  \frac{N-1}{N-2} E_{\beta(\rho)}(H) +  \frac{(d_0 - 1)Z_{\beta(\rho)}^{-1}\epsilon_{\max}}{N-2}  +  \mathcal{O} \left( \frac{1}{N^2} \right)\;.
\label{RESULT_inst_duelivelli}    
\end{equation}
Finally consider the situation where $S(\rho) <\ln {d_0}$ which even prevents us  the possibility of identifying a Gibbs counterpart of $\rho$.
Here -as shown in Sec.~\ref{SECIVC}-
Eq.~(\ref{RESULT_inst}) can be replaced by 
\begin{equation}\label{HOPEFULLYLAST} 
E(\rho; H)  < \epsilon_{\max} (d-d_0) \exp\left[ -N \ln d_0 + (N-1) S(\rho)) \right] \;. 
\end{equation}

\subsection{Derivation of the asymptotic bound~(\ref{RESULT_inst})} \label{sec:deriv} 
In order to calculate how much the entropy of the system increases
when passing from $\rho \in {\mathfrak P}_H^{(N)}$ to its 
$1$-structurally stable counterpart $\bar{\rho} \in \bar{\mathfrak P}_H^{(N,1)}$ defined in (\ref{spianatura}),   we need to know how much the eigenvalues $\lambda_j$ of $\rho$ can be ``spread out'' around their mean value $\bar\lambda_{j}$. 
To tackle this issue, for all eigenvalues $\epsilon$ of $H$,
we find it useful to introduce 
the corresponding  minimal and maximal populations of $\rho$, i.e. 
the quantities 
\begin{eqnarray} 
\lambda^{\min}({\epsilon}) :=\min_{\epsilon_i =\epsilon} \lambda_i \;, \qquad \quad
\lambda^{\max}({\epsilon}) :=\max_{\epsilon_i =\epsilon} \lambda_i \;, \label{DEFLAMBDAMIN} 
\nonumber \\
\end{eqnarray} 
which clearly fulfil the  inequality
\begin{eqnarray} 
\lambda^{\min}({\epsilon})  \leq \bar{\lambda}_i \leq \lambda^{\max}({\epsilon})  \;, \quad \quad \forall \epsilon_i=\epsilon\;. \label{MAGGIORE}
\end{eqnarray} 
In view of the previous discussion we shall then  assume the condition 
\begin{eqnarray} \label{ddfs} 
\lambda^{\min}(0)  < Z_{\beta(\rho)}^{-1}\;,
\end{eqnarray} 
namely the negation of condition~(\ref{COND00}).
\begin{prep}
	\label{DEG_POPS}
	Given $N\geq 2$ and  $\rho \in {\mathfrak P}_H^{(N)}$  a  $N$-passive state with entropy larger than or equal to $\ln d_0$ and 
	satisfying the condition 
	(\ref{ddfs}), the following inequality hold
	\begin{equation}
	\label{DEG_POPSZij}
	\ln \lambda_j - \ln \lambda_i < - \frac{ \ln Z_{\beta(\rho)} +\ln \lambda_j}{N-1}\;,
	\end{equation}
	for all the populations $\lambda_j$ and $\lambda_i$ of $\rho$ associated with a non-zero
	energy level $\epsilon>0$ of $H$ (i.e. $\epsilon_i=\epsilon_j=\epsilon$). 
\end{prep}
\begin{proof}
	If the energy level $\epsilon$ is not degenerate (i.e. $d_\epsilon=1$)  the inequality (\ref{DEG_POPSZij}) is trivial (in the this case the left-hand-side term is null, while the
	right-hand-side is non-negative due to (\ref{ddfs}). 
	On the contrary, if $\epsilon$ is degenerate, let $\lambda_j$ and $\lambda_i$ two different
	populations of $\rho$ that are associated to it, i.e. $\epsilon_j= \epsilon_i= \epsilon$.
	Apply hence the  $N$-passivity condition~(\ref{eq1}), choosing a population set $I_N$ which contains as only non-zero term $n_j=N$, and a population set $J_N$ which contains as only non-zero terms $n_i=N-1$ and $n_\ell=1$ with $\ell \leq d_0-1$ referring to one of the eigenvalues of the ground state energy level. Simple algebra allows us to recast this result into the inequality 
	\begin{equation}
	\label{eq:deg_pops}
	\ln \lambda_j - \ln \lambda_i \leq  \frac{\ln \lambda_\ell - \ln \lambda_j}{N-1}\;,
	\end{equation}
	which leads to~(\ref{DEG_POPSZij}) when
	taking $\lambda_\ell=\lambda^{\min}({0})$, and 
	enforcing~(\ref{ddfs}). 
\end{proof}
\begin{corol}
	\label{DEG_POPS-corol}
	Given $N\geq 2$ and  $\rho \in {\mathfrak P}_H^{(N)}$  a  $N$-passive state   with entropy larger than or equal to $\ln d_0$ and
	satisfying the condition 
	(\ref{ddfs}), the following inequalities hold
	\begin{equation}
	\label{DEG_POPSZ}
	\ln \lambda^{\max}({\epsilon}) - \ln \lambda^{\min}({\epsilon}) < - \frac{ \ln Z_{\beta(\rho)} +\ln \lambda^{\max}({\epsilon})}{N-1}\;,
	\end{equation}
	for all the energy level $\epsilon>0$ of $H$.		
\end{corol} 
\begin{proof}
	The result follows from (\ref{DEG_POPSZij}) by taking
	$\lambda_j=\lambda^{\max}({\epsilon})$, $\lambda_i=\lambda^{\min}({\epsilon})$.
\end{proof} 
Inequalities (\ref{DEG_POPSZij}) and (\ref{DEG_POPSZ})
are valid only for energy levels $\epsilon$ which are not the ground state. In the case $\epsilon_i = 0$, we can enforce only a looser upper bound:
\begin{prep}
	\label{DEG_POPS0}
	Given $N\geq 2$ and  $\rho \in {\mathfrak P}_H^{(N)}$  a  $N$-passive state  with entropy larger than or equal to $\ln d_0$ and
	satisfying the condition (\ref{ddfs}), 
	there exists an eigenvalue $\epsilon_a$ of $H$ such that
	\begin{equation}
	\label{max_divario_0}
	\ln \lambda^{\max}(0) - \ln  \lambda^{\min}(0) < \frac{\beta(\rho)\epsilon_a}{N-1} \;.
	\end{equation}
\end{prep}
\begin{proof}
	Since $\rho$ and $\omega_{\beta(\rho)}$ have the same mean energy,
	there must exist at least one eigenvalue of $\rho$ (say $\lambda_a$) associated with an 
	energy level $\epsilon_a>0$ 
	for which 
	\begin{eqnarray}\label{lambdaA} 
	\lambda_a \geq \hat{\lambda}_a=Z_{\beta(\rho)}^{-1} e^{-\beta\epsilon_a}\;,\end{eqnarray} 
	(indeed if by contradiction such level would not exist then  $E_{\beta(\rho)}(H)$ will be strictly larger than $E(\rho;H)$). 
	Let then $\lambda_i$ and $\lambda_j$ two populations associated with the ground state energy level of the 
	system (i.e. $\epsilon_i=\epsilon_j=0$). 
	Apply the $N$-passivity equation (\ref{eq1}), when selecting a population set $J_N$ which contains as only non-zero term $n_i=N$, and a population set $I_N$ which contains as only non-zero terms $n_j=N-1$ and $n_a=1$ to obtain the inequality 
	\begin{equation} \label{pass_tra0e0}
	\ln \lambda_j - \ln \lambda_i \leq \frac{\ln \lambda_i - \ln \lambda_a}{N-1}\;. 
	\end{equation}
	Identifying $\lambda_j$ and $\lambda_i$ with $\lambda^{\max}(0)$ and $\lambda^{\min}(0)$ respectively,
	Eq.~(\ref{pass_tra0e0}) leads to
	\begin{equation}
	\ln \lambda^{\max}(0)- \ln \lambda^{\min}(0) \leq \frac{\ln \lambda^{\min}(0) - \ln \lambda_a}{N-1} 
	< \frac{\beta(\rho)\epsilon_a }{N-1} ,\label{ineq1} 
	\end{equation}
	the last passage following from (\ref{ddfs}) and (\ref{lambdaA}). 	
\end{proof}

We are now ready to estimate the entropy gain $\Delta S({\rho})$ for each degenerate energy level of $H$.
We treat separately the three cases of the ground state, of the excited states with a population 
$\lambda_j$ higher than the corresponding population $\hat{\lambda}_j = Z_{\beta(\rho)}^{-1}e^{-\beta\epsilon_j}$ of the 
Gibbs state $\omega_{\beta(\rho)}$, and of 
the excited states with populations $\lambda_j$ smaller than~$\hat{\lambda}_j$.    

\begin{prep}
	\label{deltaS_maggiori}
	Given $N\geq 2$ and  $\rho \in {\mathfrak P}_H^{(N)}$  a  $N$-passive state 
	with entropy larger than or equal to $\ln d_0$ and
	satisfying the condition (\ref{ddfs}), such that there exist
	a strictly positive energy level $\epsilon>0$ of $H$ for which 
	\begin{eqnarray}\label{CONDPREP8} 
	\lambda^{\max}(\epsilon) > Z_{\beta(\rho)}^{-1}e^{-{\beta(\rho)}\epsilon}\;,
	\end{eqnarray} 
	then following inequality holds true,
	\begin{equation} \label{PRIMA1} 
	\sum_{\epsilon_j = \epsilon}\lambda_j (\ln {\lambda_j} -  \ln \bar{\lambda}_j) <  \sum_{\epsilon_j = \epsilon} \lambda_j \frac{\beta \epsilon}{N-1} \;, 
	\end{equation}
	with $\bar{\lambda}_j$ the eigenvalues of $\bar{\rho}$ defined in~(\ref{spianatura}).
\end{prep}
\begin{proof} Given $\epsilon_j=\epsilon$ the following chain of inequality can be written
	\begin{eqnarray}
	\ln \lambda_j - \ln \bar{\lambda}_j		&\leq&
	\ln  \lambda^{\max}(\epsilon)- \ln  \lambda^{\min}(\epsilon) \\ &<& - \frac{ \ln Z_{\beta(\rho)} + \ln \lambda^{\max}(\epsilon)}{N-1}  < \frac{\beta\epsilon}{N-1}\;, \nonumber 
	\end{eqnarray}
	where in the first passage we used~(\ref{MAGGIORE}), in the
	second we used  {\bf Corollary \ref{DEG_POPS-corol}}, and  in the last one we used~(\ref{CONDPREP8}).  Equation~(\ref{PRIMA1}) then follows by multiplying the above
	expression by $\lambda_j$ and summing over all possible energy levels of energy equal to $\epsilon$. 
\end{proof}

\begin{prep}
	\label{deltaS_minori}
	Given $N\geq 2$ and  $\rho \in {\mathfrak P}_H^{(N)}$  a  $N$-passive state 
	with entropy larger than or equal to $\ln d_0$ and
	satisfying the condition (\ref{ddfs}), such that there exist
	a strictly positive energy level $\epsilon>0$ of $H$ for which 
	\begin{eqnarray}\label{CONDPREP9} 
	\lambda^{\max}(\epsilon) \leq
	{\hat{\lambda}(\epsilon)} :=  Z_{\beta(\rho)}^{-1}e^{-{\beta(\rho)}\epsilon}\;,
	\end{eqnarray} 
	then following inequality holds true,
	\begin{equation} \label{tuttiminori}
	\sum_{\epsilon_j = \epsilon}\lambda_j (\ln {\lambda_j} -  \ln \bar{\lambda}_j)< 
	d_{\epsilon} \hat{\lambda}(\epsilon) \frac{{\beta(\rho)} \epsilon + 1 }{N-1}\;,
	\end{equation}
	with $\bar{\lambda}_j$ the eigenvalues of $\bar{\rho}$ defined in~(\ref{spianatura}) and
	$d_\epsilon$  the degeneracy of $\epsilon$.
\end{prep}
\begin{proof}
	Given $\epsilon_j=\epsilon$,  we can write
	
	\begin{equation} 
	\ln \lambda_j - \ln \bar{\lambda}_j 		\leq
	\ln  \lambda_j- \ln  \lambda^{\min}(\epsilon)
	<- \frac{ \ln Z_{\beta(\rho)} + \ln \lambda_j}{N-1} \;,
	\end{equation} 
	where the first inequality follows from~(\ref{MAGGIORE}) and the second from 
	{\bf Proposition~\ref{DEG_POPS}} setting $\lambda_i=\lambda^{\min}(\epsilon)$
	in Eq.~(\ref{DEG_POPSZij}).	Multiplying then by $\lambda_j$ and summing over all possible
	choices of $\epsilon_j =\epsilon$ we have that
	\begin{equation} 
	\sum_{\epsilon_j=\epsilon} \lambda_j( \ln \lambda_j - \ln \bar{\lambda}_j )	<
	-  \sum_{\epsilon_j=\epsilon} \lambda_j \frac{ \ln Z_{\beta(\rho)} + \ln \lambda_j}{N-1}
	\label{prep9_stepiniziale}
	\end{equation}
	The function $f(\lambda) \equiv -\lambda\ln(Z_{\beta(\rho)}\lambda)$ is convex for $\lambda>0$:
	\begin{eqnarray}
	f(\lambda_j) &<& f(\hat{\lambda}_j) - f'(\lambda_j) \left( \hat{\lambda}_j - \lambda_j  \right) \nonumber \\
	&=& f(\hat{\lambda}_j) + \left[ 1 + \ln(\lambda_j) \right] \left( \hat{\lambda}_j - \lambda_j  \right) \nonumber \\
	&=& f(\hat{\lambda}_j) + \left[ 1 - \beta(\rho)\epsilon + \ln\left({\lambda}_j / \hat{\lambda}_j\right) \right] \left( \hat{\lambda}_j - \lambda_j  \right) \nonumber \\
	\label{lastimagiusta}
	\end{eqnarray}
	In the case in which $\ln\left({\lambda}_j / \hat{\lambda}_j\right) < \beta(\rho)\epsilon - 1$, the inequality~(\ref{CONDPREP9}) ensures that the last term in~(\ref{lastimagiusta}) is negative, meaning that $f(\lambda_j) < f(\hat{\lambda}_j)$, and so
	\begin{equation}
	-  \lambda_j \frac{ \ln Z_{\beta(\rho)} + \ln \lambda_j}{N-1} < \hat{\lambda}_j \frac{\beta(\rho) \epsilon}{N-1} \; .
	\label{prep9_caso1}
	\end{equation}
	In the case where instead $\ln\left({\lambda}_j / \hat{\lambda}_j\right) \geq \beta(\rho)\epsilon - 1$, or equivalently $\ln {\lambda}_j < -\ln Z_{\beta(\rho)} - 1$, using again~(\ref{CONDPREP9}) we can write 
	\begin{equation}
	-  \lambda_j \frac{ \ln Z_{\beta(\rho)} + \ln \lambda_j}{N-1} < \hat{\lambda}_j \frac{1}{N-1} \; .
	\label{prep9_caso2}
	\end{equation}
	The inequality~(\ref{tuttiminori}) can then be obtained combining~(\ref{prep9_caso1}) and~(\ref{prep9_caso2}), and summing over all the energy levels with $\epsilon_j = \epsilon$.
\end{proof}

\begin{prep}
	\label{deltaS_0}
	Given $N\geq 2$ and  $\rho \in {\mathfrak P}_H^{(N)}$  a  $N$-passive state 
	with entropy larger than or equal to $\ln d_0$ and
	satisfying the condition (\ref{ddfs}),
	then
	\begin{equation} \label{eq:deltaS_0}
	\sum_{\epsilon_j = 0}\lambda_j (\ln {\lambda_j} -  \ln \bar{\lambda}_j) 
	<  (d_0-1) Z_{\beta(\rho)}^{-1}e^{\frac{\beta(\rho)\epsilon_{\max}}{N-1}} \frac{\beta(\rho)\epsilon_{\max}}{N-1}
	\end{equation}
	with $\bar{\lambda}_j$ the eigenvalues of $\bar{\rho}$ defined in~(\ref{spianatura}),
	$d_0$ the degeneracy of the ground state, and $\epsilon_{\max} $  the greatest eigenvalues of $H$.
\end{prep}
\begin{proof}
	Expunging from the sum  the negative terms we can write
	\begin{equation}
	\sum_{\epsilon_j = 0} \lambda_j \left( \ln {\lambda_j} - \ln \bar{\lambda}_j \right)
	\leq  \sum_{\epsilon_j = 0, \lambda_j > \bar{\lambda}_j } \lambda_j \left( \ln {\lambda_j} - \ln \bar{\lambda}_j \right) \;,
	\label{solopositivi}
	\end{equation}
	where  the last sum 
	contains at most $d_{0} - 1$ terms, because there is at least one $\lambda_j$ smaller than the mean. Invoking hence {\bf Proposition \ref{DEG_POPS0}} twice and Eq.~(\ref{ddfs}) we arrive to
	\begin{eqnarray}
	&&\sum_{\epsilon_j = 0, \lambda_j > \bar{\lambda_i} } \lambda_j \left( \ln {\lambda_j} - \ln \bar{\lambda}_i \right) 
	< \sum_{\epsilon_j = 0, \lambda_j > \lambda_i } \lambda_j \frac{\beta(\rho)\epsilon_{\max}}{N-1} \nonumber 
	\\ \nonumber 
	&&\qquad \qquad \leq (d_0-1) 
	\lambda^{\max} (0)  \frac{\beta(\rho)\epsilon_{\max}}{N-1} 
	\\ \nonumber 
	&&\qquad \qquad < (d_0-1) 
	\lambda^{\min} (0)  e^{\frac{\beta\epsilon_{\max}}{N-1}} \frac{\beta(\rho)\epsilon_{\max}}{N-1} \\
	&&\qquad \qquad < (d_0-1) Z_{\beta(\rho)}^{-1}e^{\frac{\beta\epsilon_{\max}}{N-1}} \frac{\beta(\rho)\epsilon_{\max}}{N-1}\;, 
	\end{eqnarray}
	which replaced into (\ref{solopositivi}) yields the thesis. 	
\end{proof}
We have now all the ingredients to estimate the maximum amount of entropy that we can gain converting $\rho \in {\mathfrak P}_H^{(N)}$ into the isoenergetic and $1$-structurally stable state $\bar{\rho} \in \bar{\mathfrak P}_H^{(N,1)}$.

\begin{prep}
	Given $N\geq 2$,  $\rho \in {\mathfrak P}_H^{(N)}$  a  $N$-passive state 
	with entropy larger than or equal to $\ln d_0$ and
	satisfying the condition (\ref{ddfs}), and $\bar{\rho} \in \bar{\mathfrak P}_H^{(N,1)}$ the
	$1$-structurally stable counterpart of $\rho$
	(as defined in \ref{spianatura}),  then the 
	entropy difference (\ref{DEFINIZIONEDELTAS}) is bounded by the inequality
	\begin{equation}
	\Delta S(\rho) 	 <
	\frac{\beta(\rho)}{N-1} \Big[ E(\rho;H) +  E_{\beta(\rho)}(H)
	+(d_0 - 1)Z_{\beta(\rho)}^{-1} \epsilon_{\max} e^{\frac{\beta(\rho)\epsilon_{\max}}{N-1}} \Big] + \frac{1}{N-1}\;, 
	\label{BOUND_DELTAS}
	\end{equation}
	with $d_0$ and $\epsilon_{\max}$ the degeneracy of the ground state and the maximum eigenvalue of $H$ respectively. 
\end{prep}
\begin{proof} Observe that
	\begin{equation}
	\Delta S(\rho) = S(\bar{\rho}) - S(\rho) 
	=
	\sum_{j} (\lambda_j   \ln \lambda_j-
	\bar{\lambda}_j \ln \bar{\lambda}_j )
	= \sum_{j} \lambda_j (  \ln \lambda_j-\ln \bar{\lambda}_j )\;, \label{DELTAS} 
	\end{equation} 
	where in the second line we used the fact that for $\epsilon_j=\epsilon_i$ one has
	$\bar{\lambda}_j=\bar{\lambda}_i$ and  that 
	$\sum_{\epsilon_j=\epsilon}  \bar{\lambda}_j=\sum_{\epsilon_j=\epsilon} {\lambda}_j$.
	Combining \textbf{Propositions~\ref{deltaS_maggiori}} and \textbf{\ref{deltaS_minori}} we hence notice that the part of the sum in Eq.~(\ref{DELTAS}) that
	involves all the energy levels above the ground state can be bounded as follows
	\begin{equation}
	\label{lemucchefannomuu}
	\sum_{\epsilon_j >0} \lambda_j (  \ln \lambda_j-\ln \bar{\lambda}_j ) < \sum_{\epsilon_j >0} \frac{\left( \lambda_j + \hat{\lambda}_j \right) \beta(\rho) \epsilon_j + \hat{\lambda}_j  }{N-1}
	=  \frac{ \beta(\rho) \left( E(\rho;H) + E_{\beta(\rho)}(H) \right)  + 1 }{N-1} \;,
	\end{equation}
	where in the last line we used the definitions of $E(\rho;H)$ and $E_{\beta(\rho)}(H)$.
	On the contrary the part of the sum in Eq.~(\ref{DELTAS}) that instead 
	involves only degenerate ground states can be instead bounded
	as in Eq.~(\ref{eq:deltaS_0}) of
	\textbf{Proposition~\ref{deltaS_0}}. 	\end{proof}
Equation~(\ref{RESULT_inst}) can be finally derived by using the identity 
(\ref{RELEVANT}) which links the energy and the entropy of Gibbs states.
Accordingly, at first order in $\Delta S(\rho)$ we get
\begin{multline} \label{ENprimordine}
E_{\beta(\bar{\rho})}(H) =  E_{\beta({\rho})}(H) + \frac{\Delta S(\rho)}{\beta({\rho})} +  \mathcal{O} \left( \Delta S^2(\rho) \right) \\
< E_{\beta({\rho})}(H) + \frac{E(\rho;H)}{N-1} + \frac{E_{\beta({\rho})}(H)}{N-1} + \frac{(d_0 - 1)\epsilon_{\max}}{N-1} Z_{\beta(\rho)}^{-1} + \frac{\beta^{-1}(\rho)}{N-1} +  \mathcal{O} \left( \frac{1}{N^2} \right)\;,
\end{multline}
where we used also $e^{\frac{\beta\epsilon_{\max}}{N-1}} = 1+ \mathcal{O} \left( \frac{1}{N} \right)$.
The bound~(\ref{RESULT_inst}) is hence obtained by first using the fact that thanks 
to the property $\beta(\bar{\rho})\leq \beta(\rho)$ we have $u(\bar{\rho}) \leq u(\rho)$, and then
replacing (\ref{ENprimordine}) into the inequality~(\ref{goulish_approach}).

\subsection{Derivation of the asymptotic bound~(\ref{RESULT_inst_duelivelli})} \label{sec:deriv_duelivelli} 
We now consider the case of an Hamiltonian $H$ whose spectrum has only two distinct eigenvalues, $0$ and $\epsilon_{\max}>0$.
Here (\ref{goulish_approach1}) can be replaced by 
\begin{equation} \label{PiallaDirectlyGibbs}
E(\rho;H) = E(\bar{\rho};H)\leq  E_{{\beta}(\bar{\rho})}(H).
\end{equation}
Assume once more that the entropy of $\rho$ is larger than or equal to $\ln d_0$ so that $E_{\beta(\rho)}$ is well defined.
In order to have $E(\rho;H) > E_{\beta(\rho)}$ (which is implied by~(\ref{MINEN}) and $\rho \neq \omega_{\beta(\rho)}$), the populations of $\rho$ must necessarily satisfy the condition~(\ref{ddfs}), and also the additional condition
\begin{eqnarray} \label{lambda1magg}
\lambda^{\max}(\epsilon_{\max}) >  Z_{\beta(\rho)}^{-1} e^{-\beta(\rho)\epsilon_{\max}} \; . 
\end{eqnarray}
The validity of~(\ref{ddfs}) allows us to use {\bf Proposition~\ref{deltaS_0}}, from which it follows equation~(\ref{eq:deltaS_0}).
On the other hand, the condition~(\ref{lambda1magg}) can be identified with the condition~(\ref{CONDPREP8}) in {\bf Proposition~\ref{deltaS_maggiori}}, implying that inequality~(\ref{PRIMA1}) holds for the energy level $\epsilon_{\max}$.

Combining equations~(\ref{eq:deltaS_0}) and~(\ref{PRIMA1}), we deduce that for a two-level Hamiltonian
\begin{multline}\Delta S(\rho) = S(\bar{\rho}) - S(\rho) =\sum_{j} (\lambda_j   \ln \lambda_j-
\bar{\lambda}_j \ln \bar{\lambda}_j ) \\
= \sum_{\epsilon_j = 0} \lambda_j (  \ln \lambda_j-\ln \bar{\lambda}_j ) + \sum_{\epsilon_j = \epsilon_{\max}} \lambda_j (  \ln \lambda_j-\ln \bar{\lambda}_j ) \\
\leq \frac{d_0-1}{N-1}Z_{\beta(\rho)}^{-1} \beta(\rho) \epsilon_{\max} e^{\frac{\beta(\rho)\epsilon_{\max}}{N-1}} +   \sum_{\epsilon_j = \epsilon_{\max}} \lambda_j \frac{\beta(\rho)\epsilon_{\max}}{N-1} \\
= 	\frac{\beta(\rho)}{N-1} \Big[ E(\rho;H) + \epsilon_{\max}(d_0-1)Z_{\beta(\rho)}^{-1} e^{\frac{\beta(\rho)\epsilon_{\max}}{N-1}} \Big] \;.
\label{DELTAS_duelivelli} 
\end{multline} 
The bound~(\ref{DELTAS_duelivelli}) is similar to the bound~(\ref{BOUND_DELTAS}), but it lacks the term proportional to $E_{\beta(\rho)}$. Using~(\ref{RELEVANT}) and  $e^{\frac{\beta\epsilon_{\max}}{N-1}} = 1+ \mathcal{O} \left( \frac{1}{N} \right)$, we can convert the bound~(\ref{DELTAS_duelivelli}) on $\Delta S(\rho)$ in an asymptotic bound on $E_{\beta(\bar{\rho})}$, which is equation~(\ref{ENprimordine}) without the term proportional to $E_{\beta(\rho)}$, i.e.,
\begin{multline} \label{ENprimordine_duelivelli}
E_{\beta(\bar{\rho})}(H) =  E_{\beta({\rho})}(H) + \frac{\Delta S(\rho)}{\beta({\rho})} +  \mathcal{O} \left( \Delta S^2(\rho) \right) \\
< E_{\beta({\rho})}(H) + \frac{E(\rho;H)}{N-1} + \frac{(d_0 - 1)\epsilon_{\max}}{N-1} Z_{\beta(\rho)}^{-1} +  \mathcal{O} \left( \frac{1}{N^2} \right)\;.
\end{multline}
Replacing (\ref{ENprimordine_duelivelli}) into the inequality~(\ref{PiallaDirectlyGibbs}), we therefore obtain the bound~(\ref{RESULT_inst_duelivelli}).

\subsection{Derivation of Eq.~(\ref{HOPEFULLYLAST})} \label{SECIVC} 
Here we focus on the case where we have a too small entropy to even identify a Gibbs isoentropic counterpart, i.e. 
\begin{equation}
S(\rho) < \ln d_0 \;.  \label{DDF1} 
\end{equation}
Let $\Delta S_0(\rho)$ denote the cost of levelling up only the ground state populations of $\rho$, i.e. constructing a density matrix $\bar{\rho}_0$ 
obtained by replacing the ground state populations of $\rho$  with $
\bar{\lambda}_0=\frac{1}{d_0} \sum_{\epsilon_j=0} \lambda_j$ while leaving all the other populations untouched,  
\begin{equation} \label{definizione_deltaS0}
\Delta S_0(\rho) := S(\bar{\rho}_0) -S(\rho) =\sum_{\epsilon_j = 0} \lambda_j (\ln \lambda_j - \ln\bar\lambda_0)\;.
\end{equation}
By majorization it is easy to verify that the entropy of $\bar{\rho}_0$ is not smaller than the one of 
the Gibbs ground state $\omega_{\infty}$, therefore we can write 
\begin{equation}
\Delta S_0(\rho) \geq \ln d_0 - S(\rho)\;.
\end{equation}
Furthermore, we notice that in the present context $\hat{\lambda_a} = 0$ if $\epsilon_a > 0$, and therefore the inequality~(\ref{pass_tra0e0}) is valid for every $\lambda_i$ and $\lambda_a$ such that $\epsilon_i = 0$ and $\epsilon_a > 0$.
Using~(\ref{MAGGIORE}) into~(\ref{definizione_deltaS0}), then applying~(\ref{pass_tra0e0}) and observing that $\lambda^{\min}(0) \leq 1/d_0$, we have that
\begin{multline}
\Delta S_0(\rho) \leq 
\sum_{\epsilon_j = 0} \lambda_j (\ln \lambda_j - \ln\bar\lambda_0) \\
\leq
\sum_{\epsilon_j = 0} \lambda_j (\ln \lambda_j - \ln\lambda^{\min}(0))  
\leq \sum_{\epsilon_j = 0} \lambda_j \frac{\ln\lambda^{\min}(0) - \lambda_a}{N-1}  
\leq - \frac{\ln d_0 - \lambda_a}{N-1} \;,
\end{multline}
which implies 
\begin{equation}
\label{ilmax_chetraborda}
\epsilon_a > 0 , \lambda_a > 0 \implies - \ln \lambda_a \geq N \ln d_0 - (N-1) S(\rho)\;.\end{equation}

There are $d-d_0$ levels above the ground state, and their contribution to the total entropy is bounded by
\begin{equation}
\label{piallamag0}
- \sum_{\epsilon_a > 0, \lambda_a > 0} \lambda_a \ln \lambda_a \leq - \sum_{\epsilon_a > 0, \lambda_a > 0} \lambda_a \ln{\left( \frac{\sum_{\epsilon_a > 0, \lambda_a > 0} \lambda_a  }{d-d_0} \right) }
\end{equation}

On the other hand, exploiting (\ref{ilmax_chetraborda}) we derive 
\begin{equation}
\label{Rindfleischetikettierungsuberwachungsaufgabenubertragungsgesetz}
- \sum_{\epsilon_a > 0, \lambda_a > 0} \lambda_a \ln \lambda_a
\geq ( N \ln d_0 - (N-1) S(\rho)) \sum_{\epsilon_a > 0, \lambda_a > 0} \lambda_a
\end{equation}
From~(\ref{piallamag0}) and~(\ref{Rindfleischetikettierungsuberwachungsaufgabenubertragungsgesetz}) we deduce the inequality
\begin{equation}
( N \ln d_0 - (N-1) S(\rho)) \leq - \ln{\left( \frac{\sum_{\epsilon_a > 0, \lambda_a > 0} \lambda_a  }{d-d_0} \right) }
\end{equation}
or
\begin{equation}
\sum_{\epsilon_a > 0, \lambda_a > 0} \lambda_a \leq (d-d_0) e^{ -N \ln d_0 + (N-1) S(\rho) }
\end{equation}
which in conclusion gives us
\begin{multline}
E(\rho; H) = \sum_{\epsilon_a > 0, \lambda_a > 0} \epsilon_a \lambda_a < \epsilon_{\max}  \sum_{\epsilon_a > 0, \lambda_a > 0} \lambda_a  \\
< \epsilon_{\max} (d-d_0) \exp\left[ -N \ln d_0 + (N-1) S(\rho)) \right] \; .
\end{multline}

\section{Some considerations about commensurable spectra} \label{APPEC} 

In Sec.~\ref{subsection1} we commented about the fact that 
for two-dimensional systems ($d=2$)  the hierarchy~(\ref{ORDE1}) trivialises (the structurally stable passive states being
also $N$ passive for all $N$) due to the fact that 
all density matrices which are diagonal in the energy basis can be cast in the Gibbs form for some proper choice of $\beta$ and $Z$. 
On the contrary, as the dimensionality increases,  Eq.~(\ref{GREATRES}) implies that 
the exponential connection 
\begin{eqnarray} \label{PROP1} \lambda_i = e^{-\beta\epsilon_i}/Z\;,\end{eqnarray}
which according to Eq.~(\ref{GIBBS}) links  the energy levels and the associated populations,
is recovered only with the hypotheses of complete passivity and structural stability. 
This general rule  admits  some notable exceptions  when the spectrum of the system Hamiltonian exhibits 
special properties. 
In particular, it is possible to show that 
if a subset  of the energy levels  of $H$ are commensurable, then the associated populations of a 
state $\rho$ which is  structurally stable and $N$-passive (with $N$ sufficiently large but finite), must be expressed as in Eq.~(\ref{PROP1}) 
for some proper choice of $\beta$ and $Z$. More specifically 
\begin{prep}
	\label{preqeig}
	If $\epsilon_a < \epsilon_b < \epsilon_c$ are three energy levels of $H$ such that 
	\begin{eqnarray} \frac{\epsilon_c - \epsilon_a}{\epsilon_b - \epsilon_a} = \frac{p}{q}\;,\label{PROP2} \end{eqnarray}  for some integers $p$ and $q$, and if $N \geq p$, then in any $N$-passive, $N$-structurally stable state $\rho\in  \bar{\mathfrak P}_H^{(N,N)}$  the corresponding
	eigenvalues $\lambda_a$, $\lambda_b$ and $\lambda_c$ can be written as in Eq.~(\ref{PROP1}) for some 
	given values of $\beta,Z\geq 0$.
\end{prep}
\begin{proof}
	Equation~(\ref{PROP2}) can be equivalently expressed as 
	$p\epsilon_b = q\epsilon_c + (p-q)\epsilon_a$. 
	Then the two eigenstates $\ket{\epsilon_b}^{\otimes p} \otimes \ket{\epsilon_0}^{\otimes (N-p)}$ and $\ket{\epsilon_c}^{\otimes q} \otimes \ket{\epsilon_a}^{\otimes (p-q)} \otimes \ket{\epsilon_0}^{\otimes (N-p)}$ of $\rho^{\otimes N}$ have the same energy (notice that we are using here that since $\rho\in  \bar{\mathfrak P}_H^{(N,N)}$ it is diagonal in the energy eigenbasis). Thence according to structurally stable condition (\ref{eq11}) they must have the same populations, i.e.  
	\begin{equation} \lambda_b^p \lambda_0^{N-p} = \lambda_c^q \lambda_a^{p-q}\lambda_0^{N-p} \implies \lambda_b^p = \lambda_c^q \lambda_a^{p-q}\;, \end{equation}  
	which implies (\ref{PROP1}).
\end{proof}
\begin{corol}
	For a Hamiltonian with equally spaced energy levels ($\epsilon_n = n\epsilon_1$), there are no nontrivial 
	$N$-passive, $N$-structurally stable states for $N \geq 2$.
\end{corol}
\begin{corol}
	For a generic discrete Hamiltonian $H$ whose energy levels are commensurable, there are no nontrivial 
	$N$-passive, $N$-structurally stable states for any $N \geq N*$, where 
	\[ N^* = \lcm \bigg\{ p \bigg| \frac{\epsilon_{i+2} - \epsilon_i}{\epsilon_{i+1} - \epsilon_i}  = \frac{p}{q} \bigg\}. \]
\end{corol}
The last statement  leads us to the following observation which holds for continuous 
variable systems -- the definition of $N$-passivity being easily generalized in this case.  
\begin{corol}
	For a Hamiltonian $H$ with a purely continuous energy spectrum, there are no non-Gibbs $N$-passive, $N$-structurally stable states for $N \geq 2$.
\end{corol}
\begin{proof}
	Take any two energies $\epsilon_a < \epsilon_c$. Since the spectrum is continuous, there exist eigenstates with any possible energy between $\epsilon_a$ and $\epsilon_c$; then we can always find a suitable $\epsilon_b$ to satisfy the condition of \textbf{Proposition \ref{preqeig}}.
\end{proof}

\section{Conclusions} \label{sec:con}

We derived upper bounds for the mean energy of $N$-passive, structurally stable configurations $\rho$. We also
give inequalities that apply for $N$-passive states which are not necessarily structurally stable, in the asymptotic limit of large $N$.
Our inequalities depend on the spectral quantity $R(H)$; the latter will typically be larger for larger values of the Hilbert space dimension $d$, resulting in looser upper bounds. On the contrary, we expect that the ratio between the maximal energy of an $N$-passive state and the energy of the isoentropic Gibbs state will, in general, be smaller for larger dimensions $d$, because the eigenvalues of $\rho$ will be constrained by more conditions. In the continuum limit, as we have seen, the set of $N$-passive, $N$-structurally stable collapses on the set of Gibbs states.

Possible future development of the present approach could be the study the connection between higher momenta of the energy distribution
of $\rho$ and those of its Gibbs isoentropic counterpart $\omega_{\beta(\rho)}$. More generally one could also employ the technique we present
here for estimating how the distance between $\rho$ and $\omega_{\beta(\rho)}$ drops when $N$ increases. 

\section*{Acknowledgment}
\addcontentsline{toc}{section}{Acknowledgment}	
We thank G. M. Andolina for  comments. 

\appendix
\section{More on the ergotropy functional} \label{APPA} 
The erogotropy functional~(\ref{ERGO}) can be casted in a more compact formula by explicitly solving the
optimization over $U$. For this purpose let us write $\rho$ as 
\begin{eqnarray} \label{RHO} \rho=\sum_{j=0}^{d-1} \lambda_j |\phi_j\rangle\langle \phi_j|\;,\end{eqnarray}  
with eigenvectors 
$\{ |\phi_j\rangle\}_{j}$
and associated eigenvalues $\{ \lambda_j\}_{j}$,
which, without loss of generality, we shall assume to 
be organized
in non-increasing order, i.e.  
\begin{eqnarray} \label{ORD2ww} 
\lambda_{j+1} \leq \lambda_j  ,\qquad  \forall j\in \{ 0,\cdots, d-2 \} \;.
\end{eqnarray} 
A passive counterpart ${\rho}_p$ of $\rho$ is now identified as an element of $\mathfrak{S}$ 
which 
is diagonal with respect to the energy eigein-basis~$\{|\epsilon_j\rangle\}_j$
and which can be expressed as 
\begin{eqnarray} \label{PASS} 
{\rho}_p := \sum_{j=0}^{d-1} \lambda_j^{(\downarrow)}  |\epsilon_j\rangle\langle \epsilon_j|\;, 
\end{eqnarray} 
where $\{ \lambda_j^{(\downarrow)}\}_j$ is a relabelling of $\{ \lambda_j\}_{j}$ that fulfils the following ordering
\begin{eqnarray} 
\epsilon_i > \epsilon_j \qquad \Longrightarrow \qquad \lambda_{i}^{(\downarrow)} \leq \lambda_{j}^{(\downarrow)} \;. \label{ORD12} 
\end{eqnarray} 
In other words ${\rho}_p$ is an element of ${\mathfrak P}_H^{(1)}$ that is iso-spectral to $\rho$, i.e.
which admits the $\{ \lambda_j\}_{j}$ has eigenvalues. Accordingly there exists always a unitary transformation
${U}_p$ such that connects them, i.e. ${\rho}_p = {U}_p \rho {U}_p^\dag$. 
It should also be noticed that due to the special ordering we fixed in  (\ref{ORD2ww}) and (\ref{ORD1})
an examples of passive state~(\ref{PASS})  is given by the density matrix
\begin{eqnarray} \label{PASSCAN} 
\tilde{\rho}_p= \sum_{j=0}^{d-1} \lambda_j  |\epsilon_j\rangle\langle \epsilon_j|\;, 
\end{eqnarray} 
obtained from $\rho$ by simply replacing $|\phi_j\rangle$ with $|\epsilon_j\rangle$ for all~$j$. 
If the Hamiltonian $H$ is  explicitly not degenerate (i.e. if in Eq.~(\ref{ORD1}) is verified with strict inequalities), $\tilde{\rho}_p$ is  the unique passive counterpart of $\rho$. 
However, if $H$ instead admits some degree of degeneracy then  this is not true
and $\rho$ may admits other passive counterparts others than (\ref{PASSCAN}) which can be obtained from the
latter by means of arbitrary unitary rotations that do not mix up eigenspaces associated with different eigenvalues
(this freedom in the definition of $\rho_p$ is associated with the fact that 
indeed if $\epsilon_i=\epsilon_j$ for some $i\neq j$, then 
Eq.~(\ref{ORD12}) does not fix any relative ordering between the associated populations). 
In any case all passive counterparts of $\rho$ will have the same mean energy, i.e. 
\begin{eqnarray} 
E({\rho}_p; H) = \sum_{j=0}^{d-1}  \lambda_j^{(\downarrow)}  \epsilon_{j}  =  \sum_{j=0}^{d-1}  \lambda_j \epsilon_{j} =
E(\tilde{\rho}_p; H)\;.
\end{eqnarray} 
Most importantly one can verify that the unitaries $U$ which leads to the maximum 
in the right-hand-side of Eq.~(\ref{ERGO}) are exactly the one that maps $\rho$ into one of it passive counterparts,
accordingly we can write 
~\cite{PASSIVE1,PASSIVE2} 
\begin{eqnarray} 
{\cal E}^{(1)}(\rho; H) &=& E(\rho; H) - E({\rho}_p;H)   \\
&=&  \sum_{j,j'=0}^{d-1} \lambda_j \epsilon_{j'} (|\langle \phi_j| \epsilon_{j'}\rangle|^2 -  \delta_{j,j'})\;,
\end{eqnarray} 
with $\delta_{j,j'}$ being the Kronecker delta. 

\section{Alternative proof  of Eqs.~(\ref{GREATRES}) and
	(\ref{IDE1}). } \label{PROFGI} 
The identity~(\ref{GREATRES}) establishes that Gibbs and ground states are the only CP configurations of the system $A$,
while (\ref{IDE1}) specifies that the Gibbs are also the only CPSS density matrices. 
Explicit proofs of these  statements  can be found in Refs.~\cite{PASSIVE1,PASSIVE2,PASSIVE4,SKRZ}.
In what follows however we give a simple, alternative demonstration of this fact   based on some simple geometric considerations. 
\begin{prep}
	\label{pre1}
	A density matrix $\rho$ of $A$ is  a CP state  
	if and only if it is either an element of the  Gibbs set ${\mathfrak{G}}_H$ or an element of the ground set
	${\mathfrak{S}}_H^{(G)}$. 
\end{prep}
\begin{proof}
	Since CP states, as well as the elements of ${\mathfrak{G}}_H$ and ${\mathfrak{S}}_H^{(G)}$ ,
	are  diagonal in the energy eigenbasis, we can restrict the analysis to this
	special case assuming that our $\rho$ has the form~(\ref{RHOpas}), i.e. 
	\begin{eqnarray} \label{RHOpas1} \rho=\sum_{j=0}^{d-1} \lambda_j |\epsilon_j\rangle\langle \epsilon_j|\;.\end{eqnarray}
	Consider then condition {\it ii)} that enforces  $N$-order passivity. 
	Introducing  the positive quantities $b_i := -\ln \lambda_i$  from 
	Eq.~(\ref{eq1}) it follows that $\rho$ is CP if  and only if, for all $N$,  
	and for all allowed choices of the sets $I_N:= \{ n_1, n_2, \cdots, n_d\}$, $J_N:=  \{ m_1,m_2, \cdots, m_d\}$, 
	we have    
	\begin{equation}
	\label{eq2}
	\sum_{i=0}^d  {n_i} \epsilon_i >   \sum_{j=0}^d  {m_j}  \epsilon_j \quad \Longrightarrow \quad  
	\sum_{i=0}^d {n_i} b_i    \geq    \sum_{j=0}^d  {m_j}  b_j \;,  
	\end{equation} 
	where  the regularization (\ref{RULE1}) translates into 
	\begin{eqnarray} \label{RULE2} 
	(n=0) (b=\infty) =0\;,\end{eqnarray} 
	(notice however that we do not need to enforce an analogous regularization for opposite situation for the product 
	$(n=\infty) (\epsilon=0)$ which we leave explicitly  indeterminate).   
	If we interpret $b_i$ and $\epsilon_i$ as component of vectors in $\mathbb{R}^d$, Eq.~(\ref{eq2}) can be reframed as
	\begin{equation}
	\label{eq3}
	\vec{I}_N \cdot \vec{\epsilon} > \vec{J}_N \cdot \vec{\epsilon}   \quad   \Longrightarrow \quad \vec{I}_N \cdot \vec{b} \geq  \vec{J}_N \cdot \vec{b}\;,
	\end{equation}
	with  $\vec{I}_N$, $\vec{J}_N\in \mathbb{R}^d$ obtained by promoting the elements of   $I_N$ and $J_N$  into vectorial components respectively, i.e. 
	$\vec{I}_N:=(n_0,n_1,\cdots, n_{d-1})$ and~$\vec{J}_N := (m_0,m_1, \cdots, m_{d-1})$.
	Calling then $\mathbbm{1}$ the vector  $(1,1,\dots ,1)$ of $\mathbb{R}^d$, by construction we have that $\vec{I}_N \cdot \mathbbm{1} = \vec{J}_N \cdot \mathbbm{1} = N$, 
	implying that the vector $\vec{I}_N - \vec{J}_N$ is orthogonal to $\mathbbm{1}$, i.e.   $(\vec{I}_N - \vec{J}_N) \cdot \mathbbm{1} = 0$. Accordingly
	Eq.~(\ref{eq3}) rewrites 
	\begin{equation}
	\label{eq4}
	\vec{v}_N\cdot \vec{\epsilon} > 0 \quad   \Longrightarrow \quad \vec{v}_N\cdot \vec{b} \geq 0 \;, \hspace{1cm} \forall \vec{v}_N \in \mathcal{V}_N\;,
	\end{equation}
	where $\mathcal{V}_N := \left\{\tfrac{\vec{I}_N - \vec{J}_N}{ \norm{ \vec{I}_N - \vec{J}_N}  }\right\}$ is the subset  of $\mathbb{R}^d$
	of the allowed (normalized) vectors. For $N \to \infty$,  $\mathcal{V}_N$ tends to a limit subset  $\mathcal{V}_\infty := \bigcap_{N \ge 1} \bigcup_{j \geq N} \mathcal{V}_j$,
	and the CP requirement can be expressed as 
	\begin{equation}
	\label{eq5}
	\vec{v} \cdot \vec{\epsilon} > 0   \quad   \Longrightarrow \quad \vec{v} \cdot \vec{b} \geq  0 \;, \hspace{1cm} \forall \vec{v} \in \mathcal{V}_\infty\;.
	\end{equation}
	Since $\mathcal{V}_\infty$ is dense in the subspace of the unitary sphere which is orthogonal to $\mathbbm{1}$, Eq.~(\ref{eq5}) is possible only if, once projected into that subspace, the vectors $\vec{\epsilon}$ and $\vec{b}$ are linearly dependent by a non-negative proportionality constant $\beta\geq 0$. Projecting in the subspace perpendicular to $\mathbbm{1}$ is equivalent to add $Z\mathbbm{1}$ for some real constant $Z$. Therefore we must have  
	\begin{eqnarray}  \label{DDF} 
	\vec{b} = \beta \vec{\epsilon} + Z\mathbbm{1}\;, \end{eqnarray}   which expanded in components leads to
	\begin{eqnarray} 
	\lambda_i =  e^{-\beta\epsilon_i} /Z\;, \label{COND1} 
	\end{eqnarray} 
	which formally coincides with the request to have $\rho$ in the Gibbs form~(\ref{GIBBS}) (the value of 
	$Z$ being forced to coincide with $Z_\beta$ by normalization). 
	The only exception to (\ref{COND1}) occurs in the limiting case where 
	the identity (\ref{DDF}) is fulfilled with 
	an infinite value of $\beta$. Under this circumstance
	for all  $\epsilon_i$ which are strictly larger than zero (i.e. for all $i\geq {d_0}$)
	we have that 
	$b_i$ diverges forcing the associated $\lambda_i$ to be exactly equal to zero. 
	On the other hand when $\epsilon_i=0$ (i.e. for all $i \in \{ 0, 1,\cdots, {d_0}\}$)  the form 
	$(\beta = \infty) (\epsilon =0)$ is indeterminate -- see comment below Eq.~(\ref{RULE2}) --  and  the constraint~(\ref{COND1}) needs not to be applied leaving us the
	freedom to chose the associated values of $\lambda_i$ as we wish. 
	This leads us to identify the ground state elements as the only other possible choices for being CP, concluding the proof of Eq.~(\ref{GREATRES}). 
\end{proof}

\begin{corol}
	\label{precolor}
	Gibbs states are the only CP density matrices of the system which are $1$-structurally stable, i.e.
	${\mathfrak G}_H={\bar{\mathfrak{P}}_H^{(\infty,1)}}$.
\end{corol}
\begin{proof} 
	According to {\bf Proposition~\ref{pre1}} the only CP states are the Gibbs and the ground states elements.
	For ${d_0}>1$ however  ground states need not to fulfil the constraint~(\ref{EQUAL}) 
	required for being $1$-structural stable,
	on the contrary Gibbs density matrices have $\lambda_i =  e^{-\beta\epsilon_i} /Z$ which naturally implement
	such requirement.	\end{proof}
More generally the Gibbs states verify also the stronger
requirement (\ref{eq11}) for any value of $k$, hence they are also $k$-structurally stable at all orders:
\begin{corol}
	\label{precolor1}
	Gibbs states are the only CP density matrices of the system which are structurally stable at all
	orders, i.e. 
	${\mathfrak G}_H={\bar{\mathfrak{P}}_H^{(\infty,\infty)}}$.
\end{corol}

\section{Majorization argument} \label{maj}

Here we present an explicit proof of the majorization argument 
used in the proof of {\bf Proposition \ref{prepFIN}}, i.e. we show that 
if  
\begin{eqnarray} \lambda_0 < \hat{\lambda}_0\;, \label{FF} \end{eqnarray}  then   there must exist  must exist $\epsilon_{c}> \epsilon_b > 0$ such that 
\begin{eqnarray} \lambda_b \geq  \hat{\lambda}_b\;, \qquad \lambda_{c} \leq  \hat{\lambda}_{c}\;, \label{GG} 
\end{eqnarray} 
where
$\hat{\lambda}_j= Z_{\beta(\rho)}^{-1} e^{ - \beta(\rho) \epsilon_j}$ are the eigenvalues of 
the Gibbs state $\omega_{\beta(\rho)}$.

For the sake of completeness we briefly recall that given 
two probability sets $P:=\{ p_j\}_{j=1,\cdots,d}$ and $Q:=\{ q_j\}_{j=1,\cdots,d}$ whose elements 
are labelled in non-decreasing order, i.e.  $p_{j} \geq p_{j+1}$, $q_{j} \geq q_{j+1}$ for all $j\in \{ 1,\cdots, d-1\}$,
one say that $Q$ majorizes $P$ when~\cite{MAJ0,MAJ} 
\begin{eqnarray} \label{major} 
\sum_{j=1}^k q_j \geq \sum_{j=1}^k p_j \;,\qquad  \forall k\leq d-1\;,\end{eqnarray} 
the inequality being always saturated with an identity for $k=d$ due to normalization conditions. 
Furthermore if there exists at least one value  $k \leq d-1$, for which (\ref{major}) is fulfilled with a strict inequality
we say that $Q$ strictly majorizes $P$. It turns out that majorization induces an ordering for the entropy of the two sets,
so that whenever $Q$ majorizes $P$, then the entropy of the former is always smaller than or equal to the entropy of the
latter, the inequality being strict  if the strict  majorization condition applies. 
It is hence clear that if the two probability sets have identitical entropy, then there neither
$Q$ can strictly majorize $P$, nor $Q$ can strictly majorize $P$.

Taking into account the above facts, let us now go back to the proof of the property (\ref{GG}). 
The existence of $\epsilon_b > 0$ fulfilling (\ref{GG}) can be established  from (\ref{FF}) and from the
fact that
$\rho$ and $\omega_{\beta(\rho)}$ have both trace one. 
We can further observe that one can select as  $\epsilon_b$
a level of $H$ which has not the maximum energy value; indeed, if by contradiction for all $\epsilon_j$ smaller than the maximum 
energy value of $H$,  
from (\ref{FF}) it would follow that
$\rho$ would be strictly majorized by $\omega(\rho)$, which is impossible
as the two states have the same entropy. 
Now take as $\epsilon_b$ the one which has the smallest energy.
Accordingly  for all $\epsilon_i < \epsilon_b$ we have $\lambda_i \leq  \hat{\lambda}_i$ and hence,
\begin{eqnarray}  \label{COND111} 
\sum_{j=0}^{b-1} \lambda_i < \sum_{j=0}^{b-1} \hat{\lambda}_i \;, \end{eqnarray} 
the strict inequality being a consequence of (\ref{FF}). 
Therefore  there must  exist
$c'\in\{  b,\cdots, d-1\}$ such that 
\begin{eqnarray}  \label{COND111adfa} 
\sum_{j=0}^{c'} \lambda_i > \sum_{j=0}^{c'} \hat{\lambda}_i \;, \end{eqnarray} 
otherwise  $\{\lambda_i\}_i$ would strictly majorize $\{ \hat{\lambda}_i\}_i$ and the two could not have the same entropy. 
Observe then that the normalization conditions impose that 
\begin{eqnarray}  \label{COND111asdfd} 
\sum_{j=c+1}^{d} \lambda_i \leq \sum_{j=c+1}^{d} \hat{\lambda}_i \;, \end{eqnarray} 
which can only be satisfied if there exist $\epsilon_{c}\geq \epsilon_{c'+1}>\epsilon_b$ such that  $\lambda_{c} \leq  \hat{\lambda}_{c}$, 
hence proving the thesis.


\begin{thebibliography}{10}
	\providecommand{\url}[1]{#1}
	\bibitem{ALICKIREV} R.~Alicki and R.~Kosloff, ``Introduction to Quantum Thermodynamics: History and Prospects,''in \emph{Thermodynamics in the quantum regime - recent progress and outlook},
	F.~Binder, L.~A.~Correa, C.~Gogolin, J.~Anders, and G.~Adesso, Eds., Berlin, Germany: Springer, 2018. doi: \href{http://doi.org/10.1007/978-3-319-99046-0_1}{10.1007/978-3-319-99046-0\_1}
	\bibitem{PASSIVE3} A.~E.~Allahverdyan, R.~Balian, and Th.~M.~Nieuwenhuizen, ``Maximal work extraction from quantum systems,'' \emph{Europhysics Letters}, vol. 67, no.~4, pp. 565--571, Aug. 2004.
	doi: \href{http://doi.org/10.1209/epl/i2004-10101-2}{10.1209/epl/i2004-10101-2}
	\bibitem{PASSIVE1} W.~Pusz and S.~L.~Woronowicz, ``Passive states and KMS states for general quantum systems,'' \emph{Communications in
		Mathematical Physics}, vol. 58, no.~3, pp. 273--290, Oct. 1978.
	doi: \href{http://doi.org/10.1007/BF01614224}{10.1007/BF01614224}
	\bibitem{PASSIVE2} A.~Lenard, ``Thermodynamical proof of the gibbs formula for elementary
	quantum systems'', \emph{Journal of Statistical Physics}, vol. 19, no.~6, pp. 575--586, Dec. 1978.
		doi: \href{http://doi.org/10.1007/BF01011769}{10.1007/BF01011769}
	\bibitem{KUBO} R.~Kubo, ``Statistical-Mechanical Theory of Irreversible Processes. I. General Theory and Simple Applications to Magnetic and Conduction Problems,'' \emph{Journal of the Physical Society of Japan}, vol. 12, no.~6, pp. 570--586, Jun. 1957.
	doi: \href{http://doi.org/10.1143/JPSJ.12.570}{10.1143/JPSJ.12.570}
	\bibitem{MARSCH} P.~C.~Martin and J.~Schwinger, ``Theory of Many-Particle Systems. I,'' \emph{Physical Review}, vol. 115, no.~6, pp. 1342--1373, Sep. 1959.
	doi: \href{http://doi.org/10.1103/PhysRev.115.1342}{10.1103/PhysRev.115.1342}
	\bibitem{BATTY} C.~J.~K.~Batty, ``The KMS condition and passive states,'' \emph{Journal of Functional Analysis}, vol. 46, no.~2, pp. 246--257, Apr. 1982.
	doi: \href{http://doi.org/10.1016/0022-1236(82)90038-6}{10.1016/0022-1236(82)90038-6} 
	\bibitem{PASSIVE4} R.~Alicki and M.~Fannes, ``Entanglement boost for extractable work from ensembles of quantum batteries,'' \emph{Physical Review E}, vol. 87, no.~4, Apr. 2013, Art. no. 042123.
	doi: \href{http://doi.org/10.1103/PhysRevE.87.042123}{10.1103/PhysRevE.87.042123} 
	
	
	
	\bibitem{REVQTHERMO} J.~Gemmer, M.~Michel, and G.~Mahler, \emph{Quantum Thermodynamics}. 
	Berlin, Germany: Springer, 2008.
	doi: \href{http://doi.org/10.1007/b98082}{10.1007/b98082} 
	
	\bibitem{LEWENSTEINrew} M. N. Bera, A. Winter, and M. Lewenstein, ``Thermodynamics from information,''
	in \emph{Thermodynamics in the quantum regime - recent progress and outlook},
	F.~Binder, L.~A.~Correa, C.~Gogolin, J.~Anders, and G.~Adesso, Eds., Berlin, Germany: Springer, 2018.
doi: \href{http://doi.org/10.1007/978-3-319-99046-0_33}{10.1007/978-3-319-99046-0\_33}
	
	\bibitem{JANET} J.~Anders and V.~Giovannetti, ``Thermodynamics of discrete quantum processes'', \emph{New Journal of Physics}, vol. 15, Mar. 2013, Art. no. 033022.
	doi: \href{http://doi.org/10.1088/1367-2630/15/3/033022}{10.1088/1367-2630/15/3/033022}
	\bibitem{SPARC} C.~Sparciari, D.~Jennings, and J.~Oppenheim, ``Energetic instability of passive states in thermodynamics,'' \emph{Nature Communications}, vol. 8, Jan. 2017, Art. no. 1895.
	doi: \href{http://doi.org/10.1038/s41467-017-01505-4}{10.1038/s41467-017-01505-4}
	\bibitem{BRANDAO} F.~G.~S.~L.~Brad\~ao, {\it et al.}, ``The second laws of quantum thermodynamics,'' \emph{Proceedings of the National Academy of Sciences of the United States of America}, vol. 112, no.~11, pp. 3275--3279, Mar. 2015. 
doi: \href{http://doi.org/10.1073/pnas.1411728112}{10.1073/pnas.1411728112 }
	\bibitem{SKRZ} P.~Skrzypczyk, R.~Silva, and N.~Brunner, ``Passivity, complete passivity, and virtual temperatures,''
	 \emph{Physical Review E}, vol. 91, May 2015, Art. no. 052133.
	 doi: \href{http://doi.org/10.1103/PhysRevE.91.052133}{10.1103/PhysRevE.91.052133}
	\bibitem{ACIN1} K.~V.~Hovhannisyan, {\it et al.}, ``Entanglement Generation is Not Necessary for Optimal Work Extraction'', \emph{Physical Review Letters}, vol. 111, Apr. 2013, Art. no. 240401.
	doi: \href{http://doi.org/10.1103/PhysRevLett.111.240401}{10.1103/PhysRevLett.111.240401}
	\bibitem{ACIN2} M.~Perarnau-Llobet, {\it et al.}, ``Extractable Work from Correlations'', \emph{Physical Review X}, vol. 5, Oct. 2015, Art. no. 041011. 
	doi: \href{http://doi.org/10.1103/PhysRevX.5.041011}{10.1103/PhysRevX.5.041011}
	\bibitem{ERGO1} G.~Francica {\it et al.}, ``Daemonic ergotropy: enhanced work extraction from quantum correlations,'' \emph{npj Quantum Information}, vol. {\bf 3}, Mar. 2017, Art. no. 12.
	doi: \href{http://doi.org/10.1038/s41534-017-0012-8}{10.1038/s41534-017-0012-8}
	\bibitem{ERGO_BIP} M.~Alimuddin, T.~Guha, and P.~Parashar, ``Bound on Ergotropic Gap for Bipartite Separable States,'' \emph{Physical Review A}, vol. 99, May 2019, Art. no. 052320.
	doi: \href{http://doi.org/10.1103/PhysRevA.99.052320}{10.1103/PhysRevA.99.052320}
	\bibitem{BERA1} M.~N.~Bera, {\it et al.}, ``Thermodynamics as a Consequence of Information Conservation,'' \emph{Quantum}, vol. 3, Feb. 2019, Art. no. 121.
	doi: \href{http://doi.org/10.22331/q-2019-02-14-121}{10.1103/10.22331/q-2019-02-14-121}
	
	\bibitem{KURI1} D.~Gelbwaser-Klimovsky, R.~Alicki, and G.~Kurizki, ``Work and energy gain of heat-pumped quantized amplifiers,'' \emph{Europhysics Letters}, vol. 103,
	no.~6, Oct. 2013, Art. no. 60005. 
	doi: \href{http://doi.org/10.1209/0295-5075/103/60005}{10.1209/0295-5075/103/60005}
	\bibitem{KURI2} D.~Gelbwaser-Klimovsky and G.~Kurizki, ``Heat-machine control by quantum-state preparation: From quantum engines to refrigerators,'' \emph{Physical Review E}, vol. 90, no.~2, Aug. 2014, 022102.
	doi: \href{http://doi.org/10.1103/PhysRevE.90.022102}{10.1103/PhysRevE.90.022102}
	\bibitem{KURI3} D.~Gelbwaser-Klimovsky and G.~Kurizki, ``Work extraction from heat-powered quantized optomechanical setups,'' \emph{Scientific Reports}, vol. 5, Jan. 2015, Art. no.~7809. 
	doi: \href{http://doi.org/10.1038/srep07809}{10.1038/srep07809}
	\bibitem{FRIIS} N.~Friis and M.~Huber, ``Precision and Work Fluctuations in Gaussian Battery Charging,'' \emph{Quantum}, vol. 2, Apr. 2018, Art. no. 61.
	doi: \href{http://doi.org/10.22331/q-2018-04-23-61}{10.22331/q-2018-04-23-61}
	\bibitem{ACIN} M. Perarnau-Llobet, {\it et al.}, ``Most energetic passive states,'' \emph{Physical Review E}, vol. 92, no.~4, Oct. 2015, Art. no. 042147.
	doi: \href{http://doi.org/10.1103/PhysRevE.92.042147}{10.1103/PhysRevE.92.042147} 
	\bibitem{BINDER} F.~Binder, {\it et al.}, ``Quantum thermodynamics of general quantum processes,'' \emph{Physical Review E}, vol. 91, no.~3, Mar. 2015, Art. no. 032119.
	doi: \href{http://doi.org/10.1103/PhysRevE.91.032119 }{10.1103/PhysRevE.91.032119 } 
	
	
	
	\bibitem{BATTERYREV} F. Campaioli, F. A. Pollock, and S. Vinjanampathy, ``Quantum Batteries - Review Chapter,''
	in \emph{Thermodynamics in the quantum regime - recent progress and outlook},
	F.~Binder, L.~A.~Correa, C.~Gogolin, J.~Anders, and G.~Adesso, Eds., Berlin, Germany: Springer, 2018.	
	 doi: \href{http://doi.org/10.1007/978-3-319-99046-0_8}{10.1007/978-3-319-99046-0\_8}	
	\bibitem{BATTERY1} G.~M.~Andolina, ``Extractable Work, the Role of Correlations, and Asymptotic Freedom in Quantum Batteries,'' {\it et al.} 
	\emph{Physical Review Letters}, vol. 122, no.~4, Jul. 2019, Art. no. 047702.
	 doi: \href{http://doi.org/10.1103/PhysRevLett.122.047702 }{10.1007/10.1103/PhysRevLett.122.047702}
	\bibitem{BATTERY2} D.~Farina, {\it et al.}, ``Charger-mediated energy transfer for quantum batteries: An open-system approach,'' 
	\emph{Physical Review B}, vol. 99, no.~3, Jan. 2019, Art. no. 035421.
	doi: \href{http://doi.org/10.1103/PhysRevB.99.035421}{10.1103/PhysRevB.99.035421}
	\bibitem{BATTERY3} D.~Rossini, 
	G.~M.~Andolina, and M.~Polini, ``Many-body localized quantum batteries,'', \emph{Physical Review B}, vol. 100, no.~11, Sep. 2019, Art. no. 115142.
	doi: \href{http://doi.org/10.1103/PhysRevB.100.115142}{10.1103/PhysRevB.100.115142}
	\bibitem{BATTERY4} K.~Sen and U.~Sen, ``Local passivity and entanglement in shared quantum batteries,'' 2019, arXiv: \href{https://arxiv.org/abs/1911.05540}{1911.05540}

	
	
	\bibitem{GIACOMO3} G.~De Palma, ``The Wehrl entropy has Gaussian optimizers,'' \emph{Letters in Mathematical Physics}, vol. 108, no.~1, pp. 97--116, Jan. 2018.
	doi: \href{http://10.1007/s11005-017-0994-3}{10.1007/s11005-017-0994-3}
	\bibitem{GIACOMO2} G.~De Palma, {\it et al. }, ``Passive states as optimal inputs for single-jump lossy quantum channels,'' \emph{Physical Review A}, vol. 93, no.~6, Jun. 2016, Art. no. 062328.
	doi: \href{http:/10.1103/PhysRevA.93.062328}{10.1103/PhysRevA.93.062328}
	\bibitem{GIACOMO1} G.~De Palma, D.~Trevisan, and V.~Giovannetti, ``Passive States Optimize the Output of Bosonic Gaussian Quantum Channels,'' \emph{IEEE Transations on Information Theory}, vol. 62, no.~5, pp. 2895--2906, May 2016.
	doi: \href{http:/10.1109/TIT.2016.2547426}{10.1109/TIT.2016.2547426}
	
	
	\bibitem{CURVED} H.~Sahlmann and R.~Verch, ``Passivity and Microlocal Spectrum Condition,'' \emph{Communications in Mathematical Physics}, vol. 214, no.~3, pp. 705--731, Nov. 2000. 
	doi: \href{http:/10.1007/s002200000297}{10.1007/s002200000297}
	
	
	
	
	
	\bibitem{MAJ0} A.~W.~Marshall and I.~Olkin,
	\emph{Inequalities: theory of majorization and its applications}, 
	New York, NY, USA: Academic Press, 1979. 
	\bibitem{MAJ} M.~A.~Nielsen and G.~Vidal, ``Majorization and the interconversion of bipartite states'', \emph{Quantum Information and Computation}, vol. 1, no.~1, pp. 76--93, Jan. 2001.
	doi: \href{http:10.5555/2011326.2011331}{10.5555/2011326.2011331}
	
	
	
	
\end{thebibliography}
\end{document}